\newif\ifready\readyfalse
\pdfoutput=1

\documentclass{article}
\usepackage[utf8]{inputenc}
\usepackage{graphicx}
\usepackage{booktabs}
\usepackage[top=1in,left=1in,right=1in,bottom=1in]{geometry}
\usepackage{amsthm}
\usepackage{amsmath}

\usepackage{amssymb}
\usepackage[colorlinks=true,linkcolor=blue,citecolor=blue,allcolors=blue]{hyperref}
\hypersetup{
    colorlinks=true,
    linkcolor=blue,
    filecolor=blue,
    urlcolor=blue,
    allcolors=blue
}
\usepackage{relsize}
\usepackage[noend]{algpseudocode}
\usepackage{multicol}
\usepackage{tikz}
\usepackage[most]{tcolorbox}
\usepackage{color, soul}
\usepackage{pgfplots}
\usepackage{subcaption}
\pgfplotsset{
    compat=1.3,
    legend image code/.code={
        \draw [#1] (0cm,-0.1cm) rectangle (0.6cm,0.1cm);
    },
}
\usepackage{parskip}

\usepackage{thmtools}
\usepackage{thm-restate}
\usepackage{enumerate}
\usepackage{mathtools}
\usepackage{xspace}

\def\cV{\mathcal{V}}
\def\cW{\mathcal{W}}

\def\R{\mathbb{R}}

\newtheorem{theorem}{Theorem}[section]
\newtheorem*{theorem*}{Theorem}

\newtheorem{proposition}[theorem]{Proposition}
\newtheorem*{proposition*}{Proposition}
\newtheorem{lemma}[theorem]{Lemma}
\newtheorem*{lemma*}{Lemma}
\newtheorem{corollary}[theorem]{Corollary}
\newtheorem*{conjecture*}{Conjecture}

\newtheorem*{fact*}{Fact}

\newtheorem*{hypothesis*}{Hypothesis}
\newtheorem{conjecture}[theorem]{Conjecture}
\newtheorem{itheorem}[theorem]{Informal Theorem}

\newtheorem{claim}[theorem]{Claim}
\newtheorem*{claim*}{Claim}

\theoremstyle{definition}
\newtheorem{definition}[theorem]{Definition}
\newtheorem{question}[theorem]{Question}


\newcommand{\eat}[1]{}
\newcommand{\xtil}{\tilde{x}}
\newcommand{\su}[1]{^{(#1)}}

\newcommand{\N}{\mathbb{N}}

\newcommand{\calN}{\mathcal{N}}

\newcommand{\calU}{\mathcal{U}}
\newcommand{\X}{\mathcal{X}}
\newcommand{\cE}{\mathcal{E}}

\newcommand{\poly}{\mathrm{poly}}




 \newcommand{\set}[1]{\{ #1 \} }

\newcommand{\norm}[1]{\lVert #1 \rVert}

\newcommand{\Bignorm}[1]{\Big\lVert#1\Big\rVert}






\newcommand{\Iprod}[1]{\left\langle#1\right\rangle}

\newcommand{\Esymb}{\mathbb{E}}
\newcommand{\Psymb}{\mathbb{P}}

\DeclareMathOperator*{\E}{\Esymb}
 
 \DeclareMathOperator*{\ProbOp}{\Psymb}

\renewcommand{\Pr}{\ProbOp}

\newcommand{\diag}{\text{diag}}
\newcommand{\dist}{\text{dist}}




\newcommand{\rank}{\mathsf{rank}}

 \newcommand{\eps}{\varepsilon}
\renewcommand{\epsilon}{\varepsilon}



\newcommand{\calV}{\mathcal{V}}

\newcommand{\polylog}{\mathsf{polylog}}

\newcommand{\field}{\R}


\newcommand{\tU}{\widetilde{U}}
\newcommand{\tx}{\widetilde{x}}
\newcommand{\tu}{\tilde{u}}

\newcommand{\tz}{\tilde{z}}

\newcommand{\U}{\mathcal{U}}
\newcommand{\tcU}{\widetilde{\mathcal{U}}}
\newcommand{\tcS}{\widetilde{\mathcal{S}}}
\newcommand{\tcT}{\widetilde{\mathcal{T}}}

\newcommand{\veep}{\circledast}

\providecommand{\inparen}[1]{\left(#1\right)}
\providecommand{\inbrace}[1]{\left\{#1\right\}}

\providecommand{\inangle}[1]{\left\langle#1\right\rangle}

\newcommand{\Sym}{\mathsf{Sym}}

\newcommand{\colspn}{\mathrm{colspan}}

\newcommand{\Sel}{\mathrm{Perm}}

\newcommand{\Sela}{\Sel_{\text{avg}}}




\DeclareRobustCommand{\errorbox}[2][red!10]{%
    \begin{tcolorbox}[   
            left=0pt,
            right=0pt,
            top=0pt,
            bottom=0pt,
            colback=#1,
            colframe=#1,
            width=\linewidth,
            enlarge left by=0mm,
            boxsep=5pt,
            arc=0pt,outer arc=0pt,
            ]
            #2
    \end{tcolorbox}
}

\newcommand{\errornote}[1]{
    \errorbox{
        \textcolor{red}{
            \footnotesize{\bf Error notice:} {#1}{\bf  }
        }
    }
}

\newcommand{\erroremph}[1]{\textcolor{red}{\textbf{#1}}}

\newif\ifGPTdiff
\GPTdifftrue 

\RequirePackage[normalem]{ulem}
\definecolor{GPTpurple}{RGB}{128,0,160}
\definecolor{GPTlightpurple}{RGB}{190,150,205}

\makeatletter
\newsavebox{\GPTdel@box}
\newcommand{\GPTdel@math}[2]{%
    \begingroup
    \setbox\GPTdel@box=\hbox{$\m@th#1#2$}%
    \mathord{\hbox{%
        \rlap{\raise\dimexpr.5\ht\GPTdel@box-.5\dp\GPTdel@box\relax
            \hbox{\rule{\wd\GPTdel@box}{.4pt}}}%
        \box\GPTdel@box}}%
    \endgroup
}
\newcommand{\GPTdel@strikefield}{%
    \ifdim\wd\z@>\z@
        \setbox\GPTdel@box=\box\z@
        \setbox\z@=\hbox{%
            \rlap{\raise\dimexpr.5\ht\GPTdel@box-.5\dp\GPTdel@box\relax
                \hbox{\rule{\wd\GPTdel@box}{.4pt}}}%
            \box\GPTdel@box}%
    \fi
}
\newcommand{\GPTdel@mathfields}{%
    \let\GPTdel@setfield\set@field
    \def\set@field{\GPTdel@strikefield\GPTdel@setfield}%
    \let\GPTdel@gatherfield\set@gather@field
    \def\set@gather@field{\GPTdel@strikefield\GPTdel@gatherfield}%
}
\ExplSyntaxOn
\tl_new:N \l__gptdel_text_tl
\tl_new:N \l__gptdel_env_tl
\tl_new:N \l__gptdel_body_tl
\tl_new:N \l__gptdel_rows_tl
\seq_new:N \l__gptdel_rows_seq
\cs_new_protected:Npn \__gptdel_flush:
  {
    \tl_trim_spaces:N \l__gptdel_text_tl
    \tl_if_empty:NF \l__gptdel_text_tl
      { \exp_args:NV \sout \l__gptdel_text_tl \tl_clear:N \l__gptdel_text_tl }
  }
\cs_new_protected:Npn \__gptdel_scan:
  {
    \peek_meaning_remove:NTF \q_recursion_tail
      { \__gptdel_flush: \use_none_delimit_by_q_recursion_stop:w }
      {
        \peek_meaning_remove:NTF \c_space_token
          { \tl_put_right:Nn \l__gptdel_text_tl { ~ } \__gptdel_scan: }
          {
            \peek_meaning_remove:NTF \par
              { \__gptdel_flush: \par \__gptdel_scan: }
              {
                \peek_meaning:NTF \[ { \__gptdel_display:w }
                  {
                    \peek_meaning:NTF \( { \__gptdel_inline:w }
                      {
                        \peek_meaning:NTF $ { \__gptdel_dollar:w }
                          {
                            \peek_meaning:NTF \begin { \__gptdel_begin:w }
                              {
                                \peek_meaning:NTF \end { \__gptdel_end:w }
                                  {
                                    \peek_catcode:NTF \c_group_begin_token
                                      { \__gptdel_group:n } { \__gptdel_token:N }
                                  }
                              }
                          }
                      }
                  }
              }
          }
      }
  }
\cs_new_protected:Npn \__gptdel_token:N #1
  { \tl_put_right:Nn \l__gptdel_text_tl {#1} \__gptdel_scan: }
\cs_new_protected:Npn \__gptdel_group:n #1
  { \tl_put_right:Nn \l__gptdel_text_tl {{#1}} \__gptdel_scan: }
\cs_new_protected:Npn \__gptdel_inline:w \( #1 \)
  { \tl_put_right:Nn \l__gptdel_text_tl { \(#1\) } \__gptdel_scan: }
\cs_new_protected:Npn \__gptdel_display:w \[ #1 \]
  { \__gptdel_flush: \[ \GPTdel@math\displaystyle{#1} \] \__gptdel_scan: }
\cs_new_protected:Npn \__gptdel_dollar:w $
  { \peek_meaning_remove:NTF $ { \__gptdel_double:w } { \__gptdel_single:w } }
\cs_new_protected:Npn \__gptdel_single:w #1 $
  { \tl_put_right:Nn \l__gptdel_text_tl { $#1$ } \__gptdel_scan: }
\cs_new_protected:Npn \__gptdel_double:w #1 $$
  { \__gptdel_flush: \[ \GPTdel@math\displaystyle{#1} \] \__gptdel_scan: }
\cs_new_protected:Npn \__gptdel_begin:w \begin #1
  {
    \__gptdel_flush:
    \clist_if_in:nnTF
      { equation,equation*,align,align*,alignat,alignat*,gather,gather*,multline,multline* } {#1}
      {
        \tl_set:Nn \l__gptdel_env_tl {#1}
        \tl_clear:N \l__gptdel_body_tl
        \__gptdel_collect:w
      }
      { \begin{#1} \__gptdel_scan: }
  }
\cs_new_protected:Npn \__gptdel_end:w \end #1
  { \__gptdel_flush: \end{#1} \__gptdel_scan: }
\cs_new_protected:Npn \__gptdel_collect:w #1 \end #2
  {
    \tl_put_right:Nn \l__gptdel_body_tl {#1}
    \str_if_eq:eeTF {#2} { \l__gptdel_env_tl }
      {
        \exp_args:NVV \__gptdel_environment:nn \l__gptdel_env_tl \l__gptdel_body_tl
        \__gptdel_scan:
      }
      { \tl_put_right:Nn \l__gptdel_body_tl { \end{#2} } \__gptdel_collect:w }
  }
\cs_new_protected:Npn \__gptdel_environment:nn #1#2
  {
    \str_case:nnF {#1}
      {
        {equation}{ \begin{equation} \GPTdel@math\displaystyle{#2} \end{equation} }
        {equation*}{ \begin{equation*} \GPTdel@math\displaystyle{#2} \end{equation*} }
        {multline}{ \__gptdel_multline:nn {#1}{#2} }
        {multline*}{ \__gptdel_multline:nn {#1}{#2} }
      }
      { \group_begin: \GPTdel@mathfields \begin{#1} #2 \end{#1} \group_end: }
  }
\cs_new_protected:Npn \__gptdel_multline:nn #1#2
  {
    \seq_set_split:Nnn \l__gptdel_rows_seq { \\ } {#2}
    \tl_clear:N \l__gptdel_rows_tl
    \seq_map_indexed_inline:Nn \l__gptdel_rows_seq
      {
        \int_compare:nNnT {##1}>{1} { \tl_put_right:Nn \l__gptdel_rows_tl { \\ } }
        \tl_put_right:Nn \l__gptdel_rows_tl { \GPTdel@math\displaystyle{##2} }
      }
    \exp_args:NnV \__gptdel_raw_environment:nn {#1} \l__gptdel_rows_tl
  }
\cs_new_protected:Npn \__gptdel_raw_environment:nn #1#2
  { \begin{#1} #2 \end{#1} }
\NewDocumentCommand{\GPTdel}{+m}
  {
    \ifGPTdiff
      \group_begin:
      \color{GPTlightpurple}
      \mode_if_math:TF { \mathpalette\GPTdel@math{#1} }
        { \tl_clear:N \l__gptdel_text_tl \__gptdel_scan: #1 \q_recursion_tail \q_recursion_stop }
      \group_end:
    \fi
  }
\ExplSyntaxOff
\makeatother

\newif\ifnotes\notestrue

\ifnotes
\usepackage{color}
\definecolor{mygrey}{gray}{0.50}
\newcommand{\notename}[2]{{\textcolor{blue}{\footnotesize{\bf (#1:} {#2}{\bf ) }}}}

\newcommand{\enote}[1]{{\notename{Eric}{#1}}}
\newcommand{\vnote}[1]{{\notename{Vaidehi}{#1}}}
\newcommand{\bnote}[1]{{\notename{Aditya}{#1}}}
\newcommand{\anote}[1]{{\notename{Aravindan}{#1}}}
\newcommand{\tnote}[1]{{\notename{Tanmay}{#1}}}
\newcommand{\bob}[1]{{\textcolor{red}{\footnotesize{\bf (Bob:} {#1}{\bf ) }}}}
\else

\newcommand{\notename}[2]{{}}

\newcommand{\enote}[1]{}
\newcommand{\vnote}[1]{}
\newcommand{\bnote}[1]{}
\newcommand{\anote}[1]{}
\newcommand{\tnote}[1]{}
\newcommand{\bob}[1]{}

\fi

\newcommand{\vx}{\mathbf{x}}
\newcommand{\Xtil}{\widetilde{X}}
\newcommand{\calM}{\mathcal{M}}

\usepackage{thmtools}
\usepackage{thm-restate}

\usepackage{fullpage}
\usepackage{hyperref}
\hypersetup{colorlinks=true,linkcolor=blue}

\usepackage[ruled,vlined]{algorithm2e}

\title{
New Tools for Smoothed Analysis: Least Singular Value Bounds for Random Matrices with Dependent Entries
}

\author{
  Aditya Bhaskara\\
  {\small University of Utah}\\
  {\small \href{mailto:bhaskaraaditya@gmail.com}{\texttt{bhaskaraaditya@gmail.com}}} \and
  Eric Evert\\
  {\small Northwestern University}\\
  {\small \href{mailto:eric.evert@northwestern.edu}{\texttt{eric.evert@northwestern.edu}}} \and
  Anxin (Bob) Guo\\
  {\small Northwestern University} \\
  {\small \href{mailto:bobguo2023@u.northwestern.edu}{\texttt{bobguo2023@u.northwestern.edu}}} \and
  Tanmay Sinha\\
  {\small Northwestern University}\\
  {\small \href{mailto:tanmaysinha2029@u.northwestern.edu}{\texttt{tanmaysinha2029@u.northwestern.edu}}} \and
  Vaidehi Srinivas\\
  {\small Northwestern University} \\
  {\small \href{mailto:vaidehi@u.northwestern.edu}{\texttt{vaidehi@u.northwestern.edu}}} \and
  Aravindan Vijayaraghavan\\
  {\small Northwestern University}\\
  {\small \href{mailto:aravindv@northwestern.edu}{\texttt{aravindv@northwestern.edu}}}
}
\date{}

\begin{document}

\maketitle
\begin{abstract}

We develop new techniques for proving lower bounds on the least singular value of random matrices with limited randomness. The matrices we consider have entries that are given by polynomials of a few underlying base random variables. This setting captures a core technical challenge for obtaining smoothed analysis guarantees in many algorithmic settings. Least singular value bounds often involve showing strong anti-concentration inequalities that are intricate and much less understood compared to concentration (or large deviation) bounds. 

First, we introduce a general technique  
involving a hierarchical $\epsilon$-nets to prove least singular value bounds. 
Our second tool is a new statement about least singular values to reason about higher-order lifts of smoothed matrices, and the action of linear operators on them. 

Apart from getting simpler proofs of existing smoothed analysis results, we use these tools to now handle more general families of random matrices. This allows us to produce smoothed analysis guarantees in several previously open settings. These include new smoothed analysis guarantees for 
subspace clustering and certifying robust entanglement of subspaces, where prior work could only establish least singular value bounds for fully random instances or only show non-robust genericity guarantees.

\errornote{
    A previous version of this manuscript (\url{https://arxiv.org/abs/2405.01517v2}), by the authors Bhaskara, Evert, Srinivas, and Vijayaraghavan, contained an error.  This version of the manuscript corrects the error, and is able to recover the main theorem (Theorem \ref{thm:kronecker:sym}), and the applications to certifying subspace entanglement and subspace clustering. Bob Guo and Tanmay Sinha identified the gap and helped with the new proof of the claims.  
    The previous version also claimed a smoothed analysis result for a specific setting of power-sum decompositions (for sums of cubes of quadratics).  The provided proof was erroneous and this claim is \erroremph{not established}.  This corrected manuscript includes a discussion of the unproven claim.
}
    
\end{abstract}

\thispagestyle{empty}

\newpage
\thispagestyle{empty}
\tableofcontents

\newpage
\clearpage
\pagenumbering{arabic} 

\section{Introduction} \label{sec:intro}

Over the past two decades, there has been significant progress in using algebraic methods for high-dimensional statistical estimation (e.g.,~\cite{AnandkumarGHKT15}). Techniques like tensor decomposition have been used for parameter estimation in mixture models~\cite{AnandkumarHK12,BCV,GVX14}, shallow neural networks~\cite{SedghiA16,ATV21}, stochastic block models~\cite{AnandkumarGHKT15}, and more~\cite{Nikosreview}. Recently, more sophisticated decomposition methods based on tensor networks~\cite{MoitraW19},
circuit complexity~\cite{GKS2020} and algebraic geometry~\cite{GKS2020,JLV2023} have given to rise to new algorithms for many problems in high-dimensional geometry and parameter estimation. 
These algorithms start by building appropriate algebraic structures that ``encode'' the hidden parameters of interest. Then, they use the algebraic techniques described above for recovering the solution.

Unfortunately, in most of these applications, the recovery problem turns out to be NP hard in general. So the algorithms have provable recovery guarantees only under certain {\em algebraic} conditions.
Typically, these conditions can be formulated in terms of appropriately defined matrices being well-conditioned, i.e., having a non-negligible least singular value.
Furthermore, the least singular value determines the sample complexity and running time, and so it is important to obtain inverse polynomial bounds.

Now it is natural to ask: {\em do the algebraic conditions typically hold?} 
Due to NP hardness, we know there exist parameters for which the conditions do not hold. But how common or rare are such parameter settings/instances? A strong way to address this question is via the framework of smoothed analysis, developed in the seminal work of Spielman and Teng~\cite{ST, bwcabook,Teng2023Survey}. A condition is said to hold in a smoothed analysis setting if for \emph{any} instance, a small random perturbation of magnitude, say $\rho=1/n^2$, results in an instance that satisfies the condition with high probability.
Smoothed analysis guarantees show that any potential bad instance  is isolated or degenerate: most other instances in a small ball around it have good guarantees.
On the one hand, smoothed analysis gives a much stronger guarantee than \emph{average case analysis}, where one shows that the condition holds w.h.p. for a random choice of parameters from some distribution. On the other hand, it provides quantitative, robust analogs of {\em genericity} results in algebraic settings, which are needed in most algorithmic applications.

Considering the flavor of the algebraic non-degeneracy conditions,
the problem of smoothed analysis boils down to the following: {\em given a matrix $\calM$ whose entries are functions (typically polynomials) of some base variables, does randomly perturbing the variables result in $\calM$ having a non-negligible least singular value with high probability? }

This question is non-trivial even in very specialized settings, as it is a statement about anti-concentration --- a topic that is less understood in probability theory than concentration or large deviation bounds.
For example when the underlying variables form a matrix $U \in \R^{n \times m}$, the structured matrix $\calM = U \odot U=\big(u_i \otimes u_i\big)_{i \in [m]}$,\footnote{Here, $\otimes$ represents the standard tensor product or Kronecker product.} represents the Khatri-Rao product, and has been the subject of much past work~\cite{BCMV,ADMPSV18,BCPV} that developed intricate arguments specialized for this setting. Least singular value bounds of $\calM = \widetilde{U} \odot \widetilde{U}$ for randomly perturbed $\widetilde{U}$ have lead to smoothed analysis guarantees for several problems including tensor decomposition~\cite{BCMV}, recovering assemblies of neurons~\cite{ADMPSV18}, parameter estimation of latent variable models like mixtures of Gaussians~\cite{GHK}, hidden Markov models~\cite{BCPV}, independent component analysis~\cite{GVX14} and even learning shallow neural networks~\cite{ATV21}.
Another approach is to use concentration bounds to prove lower bounds on the least singular value~\cite{VershyninTensor, Ahn2016GraphMN, BHKX2022, RajendranT23Conc} for analyzing random instances; these techniques based on concentration bounds cannot handle smoothed instances. We
lack a broader toolkit that allows us to analyze more general classes of random matrices that arise in many other smoothed analysis settings of interest. 

Consider, for example
the symmetric lift of the matrix $\widetilde{U}$ represented by 
\[
\widetilde{U}^{\veep 2} \coloneqq ((\tilde{u}_i \otimes \tilde{u}_j + \tilde{u}_j \otimes \tilde{u}_i) : 1 \le i \le j \le m),
\]
where the columns (up to reshaping) give a basis for the space of all the symmetric matrices that are supported on the subspace $\tilde{U}$. Here $\veep$ denotes the symmetrized Kronecker product.

\begin{question}\label{qn:intro}
For a linear operator $\Phi$
 acting on the space of symmetric $n \times n$ matrices
 (e.g., a projection matrix), can we obtain an inverse polynomial lower bound with high probability on the least singular value of the matrix
\[ \calM = \Phi (\widetilde{U}^{\veep 2})=\Big( \Phi(\tilde{u}_i \otimes \tilde{u_j} + \tilde{u}_j \otimes \tilde{u}_i): 1 \le i \le j \le m \Big),\]
when $m \le c n $ for a sufficiently small $c \in (0,1)$?
\end{question}

The new techniques developed in this paper, to our knowledge, give the first inverse polynomial lower bound on the least singular value of $\calM$, and its higher order generalizations; see Theorem~\ref{ithm:kronecker:sym}. As it turns out, this already captures the Khatri-Rao product $\widetilde{U} \odot \widetilde{U}$ setting as a special case by setting $m=1$ and $\Phi$ appropriately.
One interpretation of the statement is that \(\widetilde{U} \veep \widetilde{U}\) acts like ``truly random'' subspace in the lifted space  $\Sym(\R^{n} \otimes \R^{n})$
with the same dimension.  With high probability, a random subspace of $\Sym(\R^n \otimes \R^n)$\footnote{$\Sym(\R^n \otimes \R^n)$ is the space of all symmetric $n \times n$ matrices.} with dimension $o(n^2)$ will not contain any vector near the kernel of $\Phi$. The affirmative answer to the above question shows that the lifted space that corresponds to column space of $(\widetilde{U})^{\veep 2}$ behaves similarly and is far from the kernel of $\Phi$! In other words, it is rotationally well-spread; it is not too aligned with any specific subspace.  Note that $\widetilde{U}$ only has about $nm$ truly independent coordinates or ``bits'', whereas a random subspace of the same dimension has $c \cdot n^2 m^2$ independent coordinates. Hence the lift $\calU^{\veep 2}$ of a smoothed subspace $\calU$ acts ``pseudorandom'' -- it acts like a random subspace in the lifted space with respect to all linear operators of reasonable rank. 

Matrices of this flavor arise in open questions about the smoothed analysis of various algebraic algorithms for problems like robust certification of quantum entanglement in subspaces, certifying distance from varieties~\cite{JLV2023}, and decomposition into sums of powers of polynomials~\cite{GKS2020, BHKX2022}.
Specifically, rank-$1$ matrices (of unit norm) correspond to separable or non-entangled states in bipartite quantum systems. For a certain specific choice of $\Phi$, the positive resolution of Question~\ref{qn:intro} certifies that a smoothed subspace of $n_1 \times n_2$ matrices of dimension $c n_1 n_2$ (for some $c>0$) is far from any rank-$1$ matrix of unit norm.
Moreover, in the recent algebraic algorithms of \cite{GKS2020, BHKX2022}, they consider generic or random subspaces $\U_1, \U_2, \dots, \U_t \subset \R^{n}$ and they need to argue that the corresponding $d$th order lifts $\U_1^{\veep d}, \U_2^{\veep d}, \dots, \U_t^{\veep d}$ are far from each other.

Our results give a novel and modular way to analyze such matrices. Our contributions are two fold:
\begin{itemize}
    \item We give new tools for proving least singular value lower bounds via $\eps$-nets. This involves identifying a key property that is sufficient for carrying forth net based arguments, and giving a new tool for proving such a property.
    \item We consider higher-order lifts of smoothed matrices and linear operators applied to them, as considered in Question~\ref{qn:intro}. We prove lower bounds on the least singular value of these higher-order lifts. This is our main technical result and this helps resolve open questions in smoothed analysis raised in~\cite{GKS2020,BHKX2022,JLV2023}.
\end{itemize}

\subsection{Our Results -- New Tools for Least Singular Value Lower Bounds}\label{sec:results}

\subsubsection{Hierarchical Nets}
Our first set of results focus on $\eps$-net based arguments for proving bounds for least singular values. Suppose we have a random matrix $\calM$, the idea is to consider a fixed ``test'' vector $\alpha$, prove that $\norm{\calM \alpha}$ is large enough with high probability, and then take a union bound over ``all possible vectors $\alpha$''. As the set of candidate $\alpha$ is infinite, the idea is to take a fine enough net over possible vectors $\alpha$.
The challenge when dealing with structured matrices (of the kind discussed above) is that for a single test vector $\alpha$, we do not obtain a sufficiently strong probability guarantee. This is because the individual columns of $\calM$ may not have ``sufficient randomness'', and since we do not know how $\alpha$ spreads its mass across columns, the bound will be weak. Our main observation is that in the matrices we consider for our application, as long as $\alpha$ is \emph{well spread}, we can obtain a much stronger bound. We refer to this as a ``combination amplifies anticoncentration'' (CAA) property of $\calM$. 

\emph{CAA Property} (Informal Definition). We say that $\calM$ has the CAA property if for every $k \geq 1$, for any unit test vector $\alpha$ that has $k$ entries of magnitude $\ge \delta$, we have that $\norm{\calM \alpha} \ge \Omega(\delta^p)$ for a fixed exponent $p\ge1$, with probability $1 - \exp(-\omega(k))$.

Formally, to capture the $\omega(k)$ term, we have a parameter $\beta$. See Definition~\ref{defn:caa-property} for details. Our first result is that for any matrix with this property, we have a bound on $\sigma_{\min}(\calM)$.
We note that in the earlier version of this paper, we defined this condition only for $p = 1$. It turns out having a weaker condition $p \geq 1$ suffices, and also makes the application to Khatri-Rao easier. 

\begin{itheorem}
    Suppose $\calM$ is a random matrix with $m$ columns and that $\calM$ satisfies the CAA property with parameter $\beta \in (0, 1]$. Then with high probability (indeed, exponentially small probability of failure), we have $\sigma_{\min} (\calM) > \poly(1/m)$. (See Theorem~\ref{thm:caa-to-sigma} for the formal statement.)
\end{itheorem}

The proof uses a novel $\eps$-net construction. Nets that use structural properties of the test vector $\alpha$ have been used in prior works in the context of proving least singular value bounds, notably in the celebrated work of Rudelson and Vershynin~\cite{RudelsonV}. In proving our result, the natural approach of constructing a hierarchy of nets based on increasing $k$ (and using some threshold $\delta$) does not work. Informally, this is because the error from ignoring terms that are slightly smaller than $\delta$ can add up significantly, causing the argument to fail. We introduce a new hierarchical construction that overcomes this problem.

The above technique shows that establishing  the CAA property for the random matrix $\calM$ of interest suffices for proving least singular value lower bounds. This can be shown via a direct argument when $\calM$ is simple, e.g., a random matrix with independent entries. However, for matrices with more structured entries, it can need a careful analysis. 
As one of our applications, in  subsection~\ref{intro:jacobianrank}, we establish the CAA property by using a new technique involving the Jacobian of some polynomial maps, and obtain alternate proofs of existing smoothed analysis results in~\cite{BCMV} and~\cite{ADMPSV18}. This technique may have other applications to proving anti-concentration of polynomials of random variables.

\subsubsection{Higher Order Lifts of Smoothed Matrices}
Next, we consider a general class of matrices that are obtained by linear operator $\Phi$ applied to higher order lifts i.e.,  taking the symmetrized Kronecker product of some $\rho$-perturbation $\widetilde{U}$ of an underlying matrix $U$. Here, the $\rho$-perturbation $\widetilde{U}$ of $U$ means that every entry of $\widetilde{U}$ is obtained by adding a small Gaussian perturbation $\calN(0,\rho^2)$ to the corresponding entry of $U$, independent of other entries.
In other words, the matrix of interest is $\calM = \Phi \widetilde{U}^{\veep d}$, where $d$ is a constant.
We can ask the question: are there conditions on $\Phi$ under which we can prove that $\sigma_{\min} (\calM)$ is large, with high probability over the perturbation? This question is interesting even in the setting when $\Phi$ is a projection matrix onto a subspace. We show that $\sigma_{\min}(\calM)$ is non-negligible with high probability, for any matrix $\Phi$ of sufficiently large rank (note that an assumption on the rank is necessary). 

This question captures a variety of settings studied previously. For example,~\cite{BCPV} studies matrices $\calM$ whose columns are tensor products of some underlying vectors (i.e., the columns have the form $u_{i_1} \otimes u_{i_2} \otimes \dots \otimes u_{i_d}$). This turns out to be a special case of our setting above. Likewise, in the works of both of~\cite{BHKX2022} and \cite{chandra2024Learning},  they consider a random matrix $\calM$ formed by concatenating the Kronecker products of a collection of underlying matrices, and the analysis of their algorithms relies on $\sigma_{\min} (\calM)$ being non-negligible. This also falls into our setting by choosing $\Phi$ appropriately (as we show in Corollary~\ref{corr:kron:blocks}).
Finally, as we discuss in our applications, the setting $\calM = \Phi \widetilde{U}^{\veep d}$ also directly appears in the work of~\cite{JLV2023}.

The following is an informal statement of our result. $\Sym_d (\R^n)$ will refer to a symmetrization of $(\R^n)^{\otimes d}$.\footnote{The latter can be viewed as having a coordinate for all ``ordered'' monomials of degree $d$ in $n$ variables (e.g., $x_i x_j$ and $x_j x_i$ correspond to different coordinates), while the former collects the terms with the same product. See Section~\ref{sec:prelims} for a formal description.} Also, as before, $\sigma_{\min}$ corresponds to right singular vectors.

\begin{itheorem}\label{ithm:kronecker:sym}
Suppose $\Phi$ be a projection matrix of rank $\delta \binom{n+d-1}{d}$ for some constant $\delta >0$, and let $U$ be any $n\times m$ matrix. Let $\widetilde{U}$ be a $\rho$-perturbation of $U$. Then as long as $m \le cn$ for some constant $c$, we have with probability $\ge 1- \exp(-\Omega(n))$,
\[ \sigma_{\min} ( \Phi \widetilde{U}^{\veep d}) \ge \poly \left( \rho, \frac{1}{n} \right). \]
(See Theorem~\ref{thm:kronecker:sym} for a formal statement.)
\end{itheorem}
Note that the above Theorem~\ref{ithm:kronecker:sym} with $d=2$ answers Question~\ref{qn:intro} affirmatively. It also proves a similar statement for $d$th order lifts of the form $\widetilde{U}^{\veep d}$. It shows that the column space of this lift is far from the nullspace of any linear operator of large rank, similar to a random subspace in the lifted space of the same dimension. This ``pseudorandom'' property is surprisingly true even though we have only $dnm$ random ``bits'' as opposed to $\Omega_d((mn)^d)$.

\subsection{Applications}

\subsubsection{Anti-concentration of a Vector of Polynomials and alternate proof of [Bhaskara, Charikar, Moitra, Vijayaraghavan]}
\label{intro:jacobianrank}

In order to apply our hierarchical nets technique, we have to prove that a random matrix $\calM$ of interest satisfies the CAA property. This requires taking a well-spread (in the sense discussed earlier) vector $\alpha$, and bounding the probability that $\norm{\calM \alpha}$ is small. In general, viewing $\calM \alpha$ as a vector whose entries are polynomials of the underlying random variables, this amounts to proving an anti-concentration statement for a vector of polynomials.
We  develop a sufficient condition for proving such results. This lets us establish the CAA property for some well-studied settings for $\calM$.  

Consider $P(x) := (p_1 (x), p_2 (x), \dots, p_N (x))$, where each $p_i$ is a polynomial of $n$ ``base'' random variables. Suppose we wish to show anti-concentration bounds for $P(\tilde{x})$, where $\tilde{x}$ is a perturbation of some $x$ (i.e., we wish to bound the probability that $P(\tilde{x})$ is within a small ball of a point $y$ is small, for all $y$).  One hope is to use a coordinate-wise bound (e.g., using known results like~\cite{KimVu}) and take the product over $1, 2, \dots, N$. It is easy to see that this is too good to be true: consider an example where $p_i$ are all equal; here having $N$ coordinates is the same as having just one. So we need a good metric for ``how different'' the polynomials $p_i$ are for a \emph{typical} $x$. We capture this notion using the Jacobian of the polynomial map $P$.  Recall that in this case, the Jacobian $J(x)$ is a matrix with one column per $p_i$, containing the vector of partial derivatives, $\nabla p_i (x)$. 

\emph{Jacobian rank property} (Informal Definition).  We say that $P(x)$ has the Jacobian rank property if for every $x$, at a perturbed point $\tilde{x}$, $J(\tilde{x})$ has at least $k$ singular values that are \emph{large enough} (where $k$ is a parameter).

We refer to Definition~\ref{defn:jacobian-rank} for the formal statement. Our result here is that this property implies anticoncentration:

\begin{itheorem}
    Suppose $P(x)$ defined as above satisfies the Jacobian rank property with parameter $k$. Then, in our applications, for a perturbation of any point $x$, we have that $\forall y$, $\Pr[ \norm{P(\tilde{x}) - y} < \eps ] < \exp(-\Omega(k))$. (Here, $\eps$ is a quantity that depends on the dimensions, $k$, the perturbation, and the singular value guarantee; see Theorem~\ref{thm:jacobian-to-antic} for the formal statement.)
\end{itheorem}

Intuitively, the Jacobian having several large singular values must result in anticoncentration (because $P(x)$ locally behaves linearly). However, the challenging aspect is that the Jacobian need not \emph{always} have many large singular values. Our assumption (Jacobian rank property) is itself made for a perturbed vector, i.e., we assume that $J(\tilde{x})$ has many high singular values with high probability. Further, the magnitude of these singular values will depend on the perturbation. Dealing with this issue turns out to be the main challenge in proving the theorem (see Theorem~\ref{thm:jacobian-to-antic} for a formal statement).

As an application of the Jacobian rank method, we re-prove the main result of~\cite{BCMV} and~\cite{ADMPSV18} when the unperturbed vectors have polynomially bounded norms. They consider random matrices $\calM$ where the $i$th column is $\tilde{u}_i \otimes \tilde{v}_i$, and $\tilde{u}_i, \tilde{v}_i$ are perturbed vectors in $\R^n$. We show that this $\calM$ satisfies the CAA property, and thus our first result (above) implies a least singular value lower bound. In order to prove the CAA property, we consider a combination of the columns $\sum_i \alpha_i (\tilde{u}_i \otimes \tilde{v}_i)$ and prove that if $\alpha$ has $k$ entries $\ge \delta$, then the Jacobian has $n\min\{k,n\}/2$ large singular values with high probability. Using our second result, we obtain a strong anticoncentration bound, thus completing the proof. This technique also lets us tackle Question~\ref{qn:intro} described above, but in what follows, we describe a different technique that also generalizes to higher orders.

\subsubsection{Certifying distance from variety and quantum entanglement} \label{sec:results:quantum}
Next, we discuss applications of our results on least singular value bounds for higher order lifts. The first such application is to the problem of certifying that a variety is ``far'' from a \emph{generic} linear subspace. As a simple motivation, suppose we have a linear subspace $\X$ of dimension $\delta n$ in $\R^n$ (assume $\delta < 1/2$). Then for a randomly $\rho$-perturbed subspace $\tcU$ of dimension $< n/2$, we can show that the two spaces have no overlap in a strong sense: every unit vector $u \in \X$ is at a distance $\Omega(\rho)$ from $\tcU$. It is natural to ask if a similar statement holds when $\X$ is an algebraic variety (as opposed to a subspace). This problem also has applications to quantum information, by instantiating with specific choices of the variety $\X$ e.g., the set of all rank-$1$ matrices of dimension $n_1 \times n_2$ (see~\cite{JLV2023} and references therein). Furthermore, we can ask if there is an efficient algorithm that can \emph{certify} that every unit vector in $\X$ is far from $\tcU$. 

We answer both these questions in the affirmative.

\begin{itheorem}
Suppose $\X \subset \R^n$ is an irreducible variety cut out by $\delta \binom{n+d-1}{d}$ homogeneous degree $d$ polynomials. There exists a $c>0$ such that for any $\rho$-perturbed subspace $\tcU$ of dimension at most $cn$, with probability $1- \exp(-\Omega(n))$, every unit vector in $\X$ has distance $\ge \poly\left( \rho, \frac{1}{n} \right)$ to $\tcU$. Further, this can be certified by an efficient algorithm. (See Theorem~\ref{thm:certifyingsections} for the formal statement.)
\end{itheorem}

The recent work of~\cite{JLV2023} gave an algorithm that we also use, but our new least singular value bounds imply the quantitative distance lower bound stated above. 
Applying this theorem with the variety of rank-$1$ matrices gives the following direct corollary.
\begin{corollary}\label{cor:X1}
There is a polynomial time algorithm that given a random $\rho$-perturbed subspace $\widetilde{\calU}$ of $n_1 \times n_2$ matrices of dimension $m \le c n_1n_2$ (for some universal constant $c>0$) certifies w.h.p. that $\widetilde{\calU}$ is at least $\poly(\rho,1/n)$ far from every rank-$1$ matrix of unit norm. 
\end{corollary}
The above theorem also has a direct implication to robustly certifying entanglement of different kinds, which we describe in Section~\ref{sec:entangled}.

\subsubsection{Decomposing sums of powers of polynomials}\label{sec:results:bafna}

\errornote{
    A previous version of this manuscript claimed to establish a smoothed analysis guarantee for the algorithm of \cite{BHKX2022} for power-sum decomposition, for the setting of cubes of quadratics.  This relied on three propositions, one of which had an erroneous proof.  Due to this error, we \erroremph{do not establish} the claim.  In this corrected manuscript, we discuss the (conjectured) proof overview and the erroneous claim in Section \ref{sec:bafna}. 
}

\subsubsection{Subspace clustering}  \label{sec:results:subspace}
Very recently, \cite{chandra2024Learning} extended the framework of \cite{GKS2020} to obtain robust analogs of their corresponding unsupervised learning algorithms.  One application of their work is to the problem of \emph{subspace clustering}.
In this problem, we see a set of points \(A \subset \mathbb{R}^n\), that admits a hidden partition into \(A_1, \dots, A_s\), such that the points in each \(A_i\) are drawn from a hidden underlying \(m\)-dimensional subspace \(W_i\).  Our goal is to recover the subspaces \(W_1, \dots, W_s\) from \(A\), assuming that the subspaces \(W_i\) along with the points in \(A\) are sufficiently perturbed.

To prove that this algorithm succeeds in the smoothed framework, \cite{chandra2024Learning} assume two conjectures.  We provide proofs of these conjectures, establishing rigorous guarantees for their method.

\begin{itheorem}
    The algorithm of \cite{chandra2024Learning} for subspace clustering admits a smoothed analysis guarantee.  (This corresponds to establishing the formal statements in Conjecture \ref{conj:SubspaceBlockKron} and Conjecture \ref{kayal-conjecture-2}.)
\end{itheorem}

In Section \ref{sec:subspace-clustering}, we outline the algorithm, and the two conjectures that they assume.  Then we address the conjectures using our techniques.  We note that this algorithm is one in a class of algebraic algorithms that \cite{chandra2024Learning} design.  It is an interesting future direction to see if these techniques can also extend to their more general framework.

\section{Proof Overview and Techniques} \label{sec:overview}
\subsection{Improved Net Analyses} \label{sec:overview:nets}

\paragraph{$\eps$-Nets and limitations.} The classic approach to proving least singular value bounds is an $\eps$-net argument. The argument proceeds by trying to prove that $\norm{\calM \alpha}$ is large for all $\alpha$ in the unit sphere. It does so by constructing a fine ``net'' over points in the sphere with the properties that (a) the net has a small number of points, and hence a union bound can establish the desired bound for points in the net, and (b) for every other point $\alpha$ in the sphere, there is a point $\alpha'$ in the net that is close enough, and hence the bound for $\alpha'$ ``translates'' to a bound for $\alpha$.  However, in settings where the columns $\Xtil_i$ of $\calM$ have ``limited randomness'', this approach cannot be applied in many parameter regimes of interest. The simplest example is one where each $\Xtil_i$ is of the form $\tilde{u}_i \otimes \tilde{u}_i$, where $\tilde{u}_i \in \R^n$ and we have around $m = n^2/4$ such vectors. In this case, (a) above causes a problem: the size of a net for unit vectors in a sphere in $\R^m$ is $\exp(m) = \exp(n^2/4)$. This is much too big for applying a union bound, since each column only has ``$n$ bits'' of randomness, so the failure probability we can obtain for a general $\alpha$ is $\exp(-n)$. For this specific example, the works~\cite{BCMV, ADMPSV18} overcome this limitation by considering alternative methods for showing least singular value bounds, not based on $\eps$-nets.

\paragraph{Main idea from Section~\ref{sec:eps-nets}.}  As described above, the limited randomness in each column $\Xtil_i$ limits the probability with which we can show that $\norm{\calM \alpha}$ is large. However, we observe that in many settings, as long as we consider an $\alpha$ that is \emph{spread out}, we can show that $\norm{\calM \alpha}$ is large with a significantly better probability. Informally, in this case, the randomness across many different columns gets ``accumulated'', thus amplifying the resulting bound. We refer to this phenomenon as \emph{combination amplifies anticoncentration} (CAA) (described informally in Section~\ref{sec:results}; see Definition~\ref{defn:caa-property}). Our first theorem states that the CAA property automatically implies a lower bound on $\sigma_{\min} (\calM)$ with high probability.

To outline the proof of the theorem, let us consider some unit vector $\alpha \in \R^m$. If $\alpha$ has say $m/2$ ``large enough'' entries, then the CAA property implies that $\norm{\calM \alpha}$ is non-negligible with probability $1-\exp(-m)$ (roughly), and so we can take a union bound over a (standard) $\eps$-net, and we would be done. However, suppose $\alpha$ had only $k$ entries that are large enough (defined as $>\delta$ for some threshold $\delta$), and $k \ll m$. In this case, the CAA property implies that $\norm{\calM \alpha} \ge c\delta^p$ with probability roughly $1-\exp(-k)$. While this is large enough to allow a union bound over just the large entries of $\alpha$ (placing a zero in the other entries), the problem is that there can be \emph{many} entries in $\alpha$ that are just slightly smaller than $\delta$. In this case, having $\norm {\calM \alpha_{\ge \delta}} \ge c\delta^p$ (where $\alpha_{\ge \delta}$ is the vector $\alpha$ restricted to the entries $\ge \delta$ in magnitude, and zeros everywhere else) does not let us conclude that $\norm{\calM \alpha} > 0$, unless $c$ is very large. Since we cannot ensure that $c$ is large, we need a different argument.

The idea will be to use the fact that our definition of the CAA comes with a slack parameter $\beta$. In particular, for $\alpha$ as above with $k$ values of magnitude $\ge \delta$, it allows us to take a union bound over $k \cdot m^{\beta}$ parameters. Thus, if we knew that there are at most $k \cdot m^\beta$ entries with magnitudes between $\delta'$ and $\delta$, for a smaller threshold $\delta'$, we can include them in the $\eps$-net. Defining $\delta'$ appropriately, we can ensure that the problem described above (where the slightly smaller entries cancel out the $\calM \alpha_{\ge \delta}$) does not occur. The problem now is when $\alpha$ has $> k \cdot m^{\beta}$ entries of magnitude between $\delta'$ and $\delta$. While this is indeed a problem for this value of $\delta$, it turns out that we can try to work with $\delta'$ instead. Now the problem can recur, but it cannot recur more than $(1/\beta)$ times (because each time, $k$ grows by an $m^{\beta}$ factor). This allows us to define a hierarchical net, which helps us identify the threshold $\delta$ for which the ratio of the number of entries $\ge \delta'$ and $\ge \delta$ is smaller than $m^{\beta}$.

By carefully bounding the sizes of all the nets and setting the thresholds appropriately, Theorem~\ref{thm:caa-to-sigma} follows.

\subsection{Jacobian Based Anticoncentration and alternate proof of [Bhaskara, Charikar, Moitra, Vijayaraghavan]}

As described in Section~\ref{intro:jacobianrank}, 
proving that a random matrix $\calM$ of interest satisfies the CAA property requires taking a well-spread vector $\alpha$, and bounding the probability that $\norm{\calM \alpha}$ is small. In general, viewing $\calM \alpha$ as a vector whose entries are polynomials of the underlying random variables, this amounts to proving an anti-concentration statement for a vector of polynomials
of the form $P(x) = (p_1(x), p_2(x), \dots, p_N(x))$ in some underlying variables $x$. The goal is to show that for every $x$, evaluating $P$ at a $\rho$-perturbed point $\tilde{x}$ gives a vector that is not too small in magnitude. (A slight generalization is to show that $P(\tilde{x})$ is not too close to any fixed $y$.)

We first observe that such a statement is not hard to prove if we know that the Jacobian $J(x)$ of $P(x)$ has many large singular values at \emph{every} $x$, and if the perturbation $\rho$ is small enough. This is because around the given point $x$, we can consider the linear approximation of $P(\tilde{x})$ given by the Jacobian. Now as long as the perturbation has a high enough projection onto the span of the corresponding singular vectors of $J(x)$, $P(\tilde{x})$ can be shown to have desired anticoncentration properties (by using the standard anticoncentration result for Gaussians). Finally, if $J(x)$ has $k$ large singular values, a random $\rho$-perturbation will have a large enough projection to the span of the singular vectors with probability $1-\exp(-k)$.

Now, in the applications we are interested in, the polynomials $P$ tend to have the Jacobian property above for ``typical'' points $x$, but not all $x$. Our main result here is to show that this property suffices. Specifically, suppose we know that for every $x$, the Jacobian at a $r$-perturbed point has $k$ singular values of magnitude $\ge \tau$ with high probability for each $r\in[\rho/2,\rho]$. Then, in order to show anticoncentration, we view the $\rho$ perturbation of $x$ as occurring in two independent steps: first perturb by $\rho \sqrt{1-z^2}$ for some parameter $z$, and then perturb by $\rho z$. The key observation is that for Gaussian perturbations, this is identical to a $\rho$ perturbation!

This gives an approach for proving anticoncentration. We use the fact that the first perturbation yields a point with sufficiently many large Jacobian singular values with high probability, and combine this with our earlier result (discussed above) to show that if $z$ is small enough, the linear approximation can indeed be used for the second perturbation, and this yields the desired anticoncentration bound. 

\emph{Applications.}  The simplest application for our framework is the setting where $\calM$ has columns being $\tilde{u}_i \otimes \tilde{v}_i$, for some $\rho$-perturbations of underlying vectors $u_i, v_i$. (This setting was studied in~\cite{BCMV, ADMPSV18} and already had applications to parameter recovery in statistical models.) Here, we can show that $\calM$ has the CAA property. To show this, we consider some combination $\sum_i \alpha_i (\tilde{u}_i \otimes \tilde{v}_i)$ with $k$ ``large'' coefficients in $\alpha$, and show that in this case, the Jacobian property holds. Specifically, we show that the Jacobian has $\Omega(n\min\{k,n\})$ large singular values with high probability. This establishes the CAA property, which in turn implies a lower bound on $\sigma_{\min}(\calM)$, assuming the unperturbed vectors have polynomially bounded norms. This gives an alternative proof of the results of the works above.

\subsection{Higher-order lifts and Structured Matrices from Kronecker Products} \label{sec:overview:kronecker}

Our second set of techniques allow us to handle structured matrices that arise from the action of a linear operator on Kronecker products, as described in Question~\ref{qn:intro}. For simplicity let us focus on the setting when $d=2$, and let $\Phi: \Sym(\R^n \otimes \R^n) \to \R^D$ be an (orthogonal) projection matrix of rank $R \ge 0.01 n^2$ acting on the space of symmetric matrices $\Sym(\R^n \otimes \R^n)$ (in general $\Phi$ can also be any linear operator of large rank). 
Let $m = cn$ for a small constant $c$ and $\widetilde{U} \in \R^{n \times m}$ be a small random $\rho$-perturbation of arbitrary matrix $U \in \R^{n \times m}$.
The columns of the matrix $\widetilde{U}^{\veep 2}$ are linearly independent with high probability, and span the symmetric lift of the column space of $\widetilde{U}$.
An arbitrary subspace of $\Sym(\R^n \otimes \R^n)$ of the same dimension may intersect non-trivially, or lie close to the kernel of $\Phi$. 
Theorem~\ref{ithm:kronecker:sym} shows that the column space of $\widetilde{U}^{\veep 2}$ for a smoothed $\widetilde{U}$ is in fact far from the kernel of $\Phi$ with high probability. 
Note that $\widetilde{U}$ only has about $nm$ truly independent coordinates or ``bits'', whereas a random subspace (matrix) of the same dimension has $c \cdot n^2 m^2$ independent coordinates.

\paragraph{Challenge with existing approaches} This setting captures many kinds of random matrices that have been studied earlier including \cite{BCMV, ADMPSV18, BCPV}. For example, \cite{BCPV} studies the setting when a fixed polynomial map $f: \R^n \to \R^k$ applied to a randomly perturbed vector $\tilde{u}_i$ to produce the $i$th column $f(\tilde{u}_i)$. It turns out to be a special case of our setting above when $m=1$. These works use the {\em leave-one-out approach} to lower bound the least singular value, where they establish that every column has a non-negligible component orthogonal to the span of the rest of the columns (see Lemma~\ref{lem:leaveoneout}). However this approach crucially relies on the columns bringing in independent randomness.\footnote{The work of \cite{BCPV} also handles some specific settings with a small overlap across columns, but these specialized ideas do not extend more generally to our setting.} This does not hold in our setting, since every column share randomness with $\Omega(m)$ other columns.  

In the recent algebraic algorithms of \cite{GKS2020, BHKX2022} for decomposing sum of powers of polynomials, the analysis of the algorithm involves analyzing the least singular value of different random matrices. One such  matrix $\calM$ is formed by concatenating the Kronecker products of a collection of underlying matrices.  This allows us to 
reason about the non-overlap or distance between the lifts of a collection of subspaces.
The work of \cite{BHKX2022} analyzed the {\em fully random} setting and proves least singular value bounds with intricate arguments involving graph matrices, matrix concentration, and other ideas. Specifically, like in \cite{VershyninTensor}, they show that $\E[\calM]$ has good least singular value, and then prove deviation bounds on the largest singular value of $\calM - \E[\calM]$ to get a bound of $\sigma_{\min}(\E[\calM]) - \norm{\calM - \E[\calM]}$. But this approach does not extend to the smoothed setting, since the underlying arbitrary matrix $U$ makes it challenging to get good bounds for  $\norm{\calM - \E[\calM]}$.   

For the smoothed case, when $d=2$, it turns out that we can use ideas similar to those described in Sections~\ref{sec:overview:nets} and \ref{sec:overview:anticonc} to show Theorem~\ref{ithm:kronecker:sym}. However, the approach runs into technical issues for larger $d$. Thus, we develop an alternate technique to analyze higher-order lifts that proves Theorem~\ref{ithm:kronecker:sym} for all constant $d$.

We prove Theorem \ref{thm:kronecker:sym} in two steps. First, we prove a new, general leave-one-out lemma for an arbitrary spanning set, and use it with a basis of rank-one symmetric tensors. This choice allows us to exploit the symmetric Kronecker form of the random matrix $\calM$, reducing the least singular value bound to estimating the distances between the test vectors (which are projected lifts of perturbed vectors) and associated random subspaces. The subspace and the vector depend on the same perturbation and are not independent, which prevents a direct application of known anti-concentration bounds.
Our second step handles this dependence by enlarging the subspace.  To lower bound the distance from the subspace, we consider a larger subspace which is independent of the test vector, but only projects out a small fraction of the rank of $\Phi$. We can then apply known anti-concentration bounds to obtain the required result. The full details can be found in Section \ref{sec:proof_of_kronecker_sym}. For simplicity of exposition, we describe our overview with $d=2$.

\paragraph{A general leave-one-out bound and rank-one frames.} The matrix $\calM$ is a linear operator acting on the space of symmetric tensors of order 2 over $\R^m$, which we denote by $\Sym^2(\R^m)$.
While its general action on the tensors can be complicated, it acts on symmetric rank-one tensors in a particularly simple way: 

$$\calM(z^{\otimes 2}) = \Phi(\tU z)^{\otimes 2}.$$

For any fixed unit vector $z$, $\tU z$ is a $\rho$-perturbed vector in $\R^n$. Thus, on rank-one inputs, the problem reduces to understanding a fixed linear projection of the lift of a perturbed vector. We can also think of the map  $z\mapsto\Phi(\tU z)^{\otimes 2}$ as a vector of homogenous degree-$2$ polynomials, evaluated at a $\rho$-perturbed point, whose underlying tensor representation has a high rank when viewed as a linear operator. Anti-concentration results for such polynomial maps were obtained in (\cite{BCPV}; also see Section \ref{sec:jacobian}). It was shown that such maps are inverse-polynomially anti-concentrated, with exponentially small failure probability. Equivalently, it was shown that the lift of a perturbed vector $\tx^{\otimes 2}$, is inverse-polynomially far from any fixed subspace of $\Sym^2(\R^n)$ of codimension at least $\Omega(n^2)$ with high probability. This gives us a way to lower bound $\norm{\calM(z^{\otimes 2})}$, for every fixed unit vector $z$. 

However, controlling these norms alone does not give us a least singular value lower bound. Intuitively, we want that the vector $\calM(z^{\otimes 2})$ is not only non-zero, but also far from the image under $\calM$ of the entire subspace orthogonal to $z^{\otimes 2}$.

We prove a general leave-one-out lemma (Lemma \ref{lem:leave_one_out_arbitrary_basis}) that formalizes this. We show that, for any linear operator $T:\cV \to \cW$ between vector spaces $\cV$ and $\cW$, for any finite set of unit vectors $S \subset \cV$ that span the input space $\cV$, we can bound the least singular value of $T$ as 
\begin{align*}
    \sigma_{\min{}}(T) \ge \frac{\tau(S)}{\sqrt{|S|}} \inf_{v \in S} \dist(T(v), T(\inangle{v}^\perp)).
\end{align*}
Here $\inangle{v}^\perp \subseteq \cV$ is the set of vectors orthogonal to $v$, and $\tau(S)$ is the least singular value of the matrix whose rows are the vectors in $S$. $\tau(S)$ measures how well the set spans the input space $\cV$. Note that Lemma \ref{lem:leave_one_out_arbitrary_basis} allows us to to pick sets $S$ of size larger than the dimension of the input space. This lets us choose test vectors adapted to the structure of $\calM$.

As described above, we want to use Lemma \ref{lem:leave_one_out_arbitrary_basis} with a frame of symmetric rank-one tensors. We provide a deterministic construction of a set of $d$-th order symmetric rank-one tensors over $\R^m$ (Lemma \ref{lem:good_basis_of_rank_ones}). Our constructed set $S$ has size at most $2^{d-1} \binom{m+d-1}{d}$, and satisfies $\tau(S) = \Omega_d(1)$. Allowing a frame with more vectors than the dimension of the space gives us a simpler construction, that achieves a dimension-independent lower bound on $\tau(S)$ while only using a constant factor more vectors.

\paragraph{Finishing the proof.} 
Using Lemma \ref{lem:leave_one_out_arbitrary_basis}, along with the constructed spanning set of symmetric rank-one tensors, it suffices to show that $\calM(z^{\otimes 2})$ is far from the subspace $\calM\left(\inangle{z^{\otimes 2}}^\perp\right)$, with high probability, for any fixed $z$. The challenge is that the vector and the subspace are not independent. We now show how the symmetric Kronecker structure in $\calM$ lets us replace this subspace by a larger one independent of the test vector similar to \cite{MSS, BCPV}, and hence lower bound the distance. 

First, observe that we can characterize the subspace in $\Sym^2(\R^m)$ orthogonal to the $z^{\otimes 2}$ as the 
span of symmetrized products having at least one mode orthogonal to $z$. The image under $\calM$ of this subspace is 
\begin{align*}
    \tcT_z := \calM\left(\inangle{z^{\otimes 2}}^{\perp}\right) = \Phi \cdot \Sym_2\left(\tU \cdot \inangle{z}^{\perp} \otimes \tcU \right)
\end{align*} where $\tcU \subseteq \R^n$ is the column span of $\tU$. Now notice that $\Phi(\tU z)^{\otimes 2}$ is independent of the restriction of $\tU$ to $\inangle{z}^\perp$, since components of Gaussian vectors along orthogonal directions are independent. The remaining mode, however, range over the full column space $\tcU$, and depends on $\tU z$. Replacing the subspace $\tcU$ along the second mode by all of $\R^n$ gives us a subspace that is independent of the vector $\tU z$. Moreover, the first factor still remains $\tU \cdot \inangle{z}^\perp$, which has dimension at most $m$. This ensures that the dimension is small enough that projecting out this space preserves most of the rank of $\Phi$.    

More formally, define the subspace $\tcS_z \supseteq \tcT_z$ as 
\begin{equation*}
    \tcS_z = \Phi \cdot \Sym_2\left(\tU \cdot \inangle{z}^\perp \otimes \R^n\right).
\end{equation*}
The subspace $\tcS_z$ has dimension at most $(m-1)n \le c_d \delta n^{2}$, and is independent of $\tU z$. Choosing $c_d$ small enough ensures that the dimension is only a small fraction of the rank of $\Phi$. We condition on $\tcS_z$ and treat it as a fixed subspace with codimension $\Omega(n^2)$ inside the image of $\Phi$, and use the anti-concentration bounds from \cite{BCPV} to bound the distance between the test vector and the subspace $\tcS_z$, with exponentially small failure probability. Since $\tcT_z \subseteq \tcS_z$, this bounds holds for the original subspace as well. Taking a union bound over all vectors in the frame and applying the general leave-one-out lemma completes the proof.

\section{Preliminaries}\label{sec:prelims}
We now introduce our basic definitions and notation
 For a matrix $U \in \R^{n \times m}$, let $\|U\|$ and $\|U\|_F$ denote the operator and Frobenius norms of $U$, respectively. Central to the paper are $\rho$-smoothed matrices. In particular, given a matrix $U \in \R^{n \times m}$, we let $\tilde{U} = U+E$ where $E \in \mathcal{N}(0,\rho^2)$. We commonly call $\tilde{U}$ a $\rho$-smoothing of $U$ or a $\rho$-perturbation of $U$. Similar notation is used for vector inputs $x=(x_1,\dots,x_n)$ to a polynomial $p:\R^n \to \R^m$. I.e., $\tilde{x} = x+\eta$ where $\eta \in \mathcal{N}(0,\rho^2)$. Thus, for example, $p(\tilde{x})$ is the evaluation of $p$ on a $\rho$-smoothed $x$.

Throughout this paper, we will be dealing with finite dimensional vector spaces over the real numbers equipped with an inner product, denoted by $\langle\cdot, \cdot\rangle$. We use the norm induced by the inner product $\norm{v} = \sqrt{\Iprod{v,v}}$. For a linear operator $T: \cV \to \cW$, we define the least singular value of $T$ as 
\begin{equation*}
    \sigma_{\min}(T) := \min_{x \in \cV, \norm{x} = 1} \norm{T(x)}.
\end{equation*}

\paragraph{Products.} We also frequently use the Kronecker product, denoted $\otimes$, and the Khatri-Rao product, denoted $\odot$. Given matrices, $A \in \R^{n \times m}$ and $B \in \R^{k \times \ell}$, the Kronecker product $A \otimes B$ is the block matrix
\[
A \otimes B = \begin{bmatrix}
    a_{11} B & \dots &  a_{1m_1} B \\
    \vdots & \ddots & \vdots \\
    a_{n_1 1} B & \dots & a_{n_1m_1} B
\end{bmatrix} \in R^{n k \times m \ell}.
\]
We let $A^{\otimes d} \in R^{n^d \times m^d}$ denote the Kroncker product of a total of $d$ copies of $A$. In the case that $m=\ell$, the Khatri-Rao product $A \odot B$ is defined by
\[
A \odot B = \begin{bmatrix}
    \uparrow &  & \uparrow \\
    a_1 \otimes b_1 & \dots & a_{m} \otimes b_{m} \\
    \downarrow &  & \downarrow
\end{bmatrix} \in \R^{n k \times m}.
\]
Here $a_j$ and $b_j$ denote the $j$th column of $A$ and $B$, respectively, and $a_j \otimes b_j$ is the Kronecker product (or simply the tensor product) of these columns. 

For vector spaces $\calU, \calV$, the tensor product space $\calU \otimes \calV = \set{ u \otimes v : u \in \calU, v \in \calV}$. When $\calU=\calV$, we also call $\calU^{\otimes 2} = \calU \otimes \calU$ a lift of the space $\calU$ (of degree/order $2$). This can also be generalized to $d$-wise products and lifts. When $\calU = \R^n$, the space $(\R^n)^{\otimes d}$ corresponds to the space of all $d$-th order tensors of dimensions $n \times n \dots \times n$. This is isomorphic to the space $\R^{n^d}$; each tensor can be flattened to form a vector in $n^d$ dimensions i.e., $(\R^n)^{\otimes d}\cong \R^{n^d}$. 

\paragraph{Symmetrized products.} We are often concerned with symmetrized versions of matrix products. To handle these, we introduce a (partially) symmetrized Kronecker product $\veep$ which is defined for tuples of matrices $(U^{(1)},\dots,U^{(d)})$ where $U^{(j)} \in \R^{n_j \times m}$. We define $U^{(1)} \veep U^{(2)} \veep \dots \veep U^{(d)} \in \R^{\Pi_{i=1}^d n_j \times \binom{m+d-1}{d}}$ to be the matrix with columns indexed by tuples $(i_1,i_2,\dots,i_d)$ with $1\leq i_1 \leq i_2 \leq \dots \leq i_d \leq m$ where the column corresponding to $(i_1,i_2,\dots,i_d)$ is 
\[
\frac{1}{|S_d|} \sum_{\pi \in S_d} u^{(1)}_{i_{\pi(1)}} \otimes u^{(2)}_{i_{\pi(2)}} \otimes \dots \otimes  u^{(d)}_{i_{\pi(d)}}.
\]
Here $S_d$ denotes the symmetric group on $[d]$ and $u^{(j)}_{i_{\pi(j)}}$ denotes the $i_{\pi(j)}$th column of $U^{(j)}$. For example, for matrices $U,V \in \R^{n \times m}$, the column of $U \veep V$ corresponding to a tuple $(i,j)$ with $i\leq j$ is 
\[
\frac{1}{2} (u_i \otimes v_j + u_j \otimes v_j).
\]
In the case that $i=j$, this reduces to $u_i \otimes v_i$. For a matrix $U \in \R^{n \times m}$, we let $U^{\veep d} \in \R^{n^d \times \binom{m+d-1}{d}}$ denote the $\veep$ product of a total of $d$ copies of $U$. The product $\veep$ can be viewed as a partially symmetrized version of the Kronecker product since all columns of $U^{\veep d}$ are symmetric with respect to the natural symmetrization of $\R^{n^d} \cong (\R^n)^{\otimes d}$. 

Along these lines, we introduce the operator $\Sym_d: \R^{n^d} \to \R^{n^d}$ which symmetrizes elements of $\R^{n^d}$ with respect to the identification $\R^{n^d} \cong (\R^n)^{\otimes d}$. With this notation, we have that 
\[
\Sym_d (U^{\veep d}) = U^{\veep d}.
\]
Moreover, the columns of the matrix $U^{\veep d}$ are precisely the \emph{unique} columns of the matrix $\Sym_d (U^{\otimes d})$.

Finally, for a vector space $\calU$, we have that $\calU^{\veep d} = \Sym_d(\calU^{\otimes d})$ is the space of symmetric $d$th tensors over the space $\calU$. We also call this the symmetric $d$th order lift of the space $\calU$.  

\paragraph{Leave-one-out distance.}
The leave-one-out distance of a matrix $U$ is a useful tool for analyzing least singular values. Given $U \in \R^{n \times m}$, define the leave-one-out distance $\ell(U)$ by 
\[
\ell(U) = \min_{i} \mathrm{dist} \left(u_i, \mathrm{Span} \{u_j : j \neq i\} \right).
\]
The least singular value of $U$ is related to the leave-one-out distance of $U$ through the following lemma \cite{RudelsonV}.

\begin{lemma}[Leave one out distance]\label{lem:leaveoneout}
Let $U \in \R^{n \times m}$. Then
\[
\frac{\ell(U)}{\sqrt{m}} \leq \sigma_{\min} (U) \leq \ell(U). 
\]
\end{lemma}
See also Lemma \ref{lemma:BlockLeaveoneOut} for a block-version of leave-one-out singular value bounds.

In our work we also encounter the Jacobian of a polynomial map. Given a vector valued function $P(x) = (p_1 (x), p_2(x), \dots, p_N (x))$ over underlying variables $x = (x_1, x_2, \dots, x_n)$, the Jacobian is defined as the $(n \times N)$ matrix of partial derivatives where the $(i,j)$th entry is $\frac{\partial p_j}{\partial x_i}$. Thus, the linear approximation of $P(x)$ around a point $x$ is simply $P(x + \eta) = P(x) + J(x)^T \eta$. 

\section{Hierarchical Nets and Anti-concentration from Jacobian Conditioning}\label{sec:jacobian-net}
In this section, we will primarily deal with a matrix $\calM$ of dimensions $N \times m$ where $m < N$. The columns will be denoted by $\Xtil_i$, and we wish to show a lower bound on $\sigma_m (\calM)$.
We will describe the finer $\eps$-net argument outlined in Section~\ref{sec:overview}. We begin with a formal definition of the CAA property.

\begin{definition}[CAA property] \label{defn:caa-property}
We say that a random matrix $\calM$ with $m$ columns has the CAA property with parameter $\beta \in (0,1]$ and fixed exponent $p\ge1$, if there exist $\lambda\ge1$ and $c\ge8/\beta$ (dependent only on $\calM$) such that for all $\delta>0$,
for all $k \geq 1$, 
for all unit test vectors $\alpha \in \R^m$ with at least $k$ coordinates of magnitude $\ge\delta$,
\[ \forall h \in (0, 1), \quad \Pr[ \norm{\calM \alpha} < \delta^p h / \lambda ] \le \exp \left( - c \min(m, k m^{\beta}) \log(1/h) \right). \]
\end{definition}

\paragraph{Remark.} We note that the condition $c \ge 8/\beta$ may seem strong; however, as we will see in applications, it is satisfied as long as $m$ is small enough compared to $N$, the number of rows of the matrix.

\subsection{Hierarchical nets}\label{sec:eps-nets}

The following shows that the CAA property implies a least singular value guarantee.

\begin{theorem}\label{thm:caa-to-sigma}
Suppose $\calM$ is a random matrix with $m$ columns and that $\calM$ satisfies the CAA property with some parameter $\beta \in (0,1]$ and fixed exponent $p\ge1$. Suppose additionally that we have the spectral norm bound $\norm{\calM} \le L$, where $L\ge1$, with probability $1-\eta$. Then with probability at least $1-\exp(-cm^{\beta}/8) - \eta$, we have
\[ \sigma_m (\calM) \ge (2mL\lambda)^{-C_{p,\beta}}, \]
where $C_{p,\beta}>0$ depends only on $p,\beta$ and $\lambda$ and $c$ come from the CAA property.
\end{theorem}

As discussed in Section~\ref{sec:overview}, the natural approach to proving such a result would be to take nets based on the sparsity of the test vector $\alpha$. In other words, if there are $k$ nonzero values of magnitude $\delta >0$, the CAA property yields a lower bound on $\norm{\calM\alpha}$ of order $\delta^p/\lambda$ (choosing $h$ to be a small constant), and we can take a union bound over a net of size $\exp(k)$. The issue with this argument is that $\alpha$ might have many other non-zero values that are \emph{slightly}  smaller than $\delta$, and these might lead to a zero singular value. Of course, in this case, we should have worked with a slightly smaller value of $\delta$, but this issue may recur, so we need a more careful argument.

The rest of this subsection will focus on proving Theorem~\ref{thm:caa-to-sigma}.

For defining the nets, we will use threshold values $\tau_1 = 1/m$, $\tau_2 = \theta^{p+1}/m^{2p}$, and so on (more generally, $\tau_{j+1}=\theta^{p+1}\tau_j^{2p}$). $\theta$ is a parameter that will be chosen appropriately; for now we simply use $\theta \in (0,1/(8m)]$.

We construct a sequence of nets $\calN_1, \calN_2, \dots, \calN_{s-1}$ as follows. The net $\calN_1$ is a set of vectors parametrized by pairs $(r_1, r_2) \in \mathbb{N}^2$, where: (a) $1 \leq r_1 \le m^{1-\beta}$, (b) $r_1+r_2 \le m^\beta r_1$. For each pair $(r_1, r_2)$, we include all the vectors whose entries are integer multiples of $\tau_2$ and that have exactly $(r_1 + r_2)$ non-zero entries, obtained from unit vectors with $r_1$ entries of magnitude in $[\tau_1,1]$ and $r_2$ entries of magnitude in $[\tau_2,\tau_1)$. To obtain these vectors, discard coordinates of magnitude below $\tau_2$ and round the remaining magnitudes up to multiples of $\tau_2$, while keeping their signs.

Thus, the number of vectors in $\calN_1$ for a single pair $(r_1, r_2)$ is bounded by:
\[ \binom{m}{r_1} \binom{m}{r_2} \left( 2\left\lceil\frac1{\tau_2}\right\rceil \right)^{r_1} \left( 2\left\lceil\frac1{\tau_2}\right\rceil \right)^{r_2} < \left( \frac{m}{\tau_2} \right)^{2(r_1+r_2)}. \]

The next net $\calN_2$ has vectors parametrized by $(r_1, r_2, r_3)\in\mathbb{N}^3$, where $r_1\ge1$, (a) $r_1+r_2 \le m^{1-\beta}$, (b) $r_1+r_2+r_3 \le m^{\beta} (r_1+r_2)$, and additionally, (c) $r_1+r_2 > m^{\beta} r_1$. For each such tuple, we include vectors (obtained again by rounding entries of unit vectors) that have exactly $(r_1 + r_2 + r_3)$ non-zero entries, and have values that are all integer multiples of $\tau_3$.

More generally, write $R_i=r_1+\cdots+r_i$. Then the vectors of $\calN_j$ will be parametrized by $(r_1, r_2, \dots, r_{j+1})\in\mathbb{N}^{j+1}$, where $r_1\ge1$, (a) $R_j \le m^{1-\beta}$, (b) $R_{j+1} \le m^\beta R_j$, and additionally, (c) for $1\le i<j$, we have $R_{i+1} > m^{\beta} R_i$. In other words, $R_{j+1}$ is the first cumulative value that does not grow by more than a factor $m^\beta$. For every such tuple, $\calN_j$ includes all vectors that have exactly $(r_1 + \dots + r_{j+1})$ non-zero entries, each of which is an integer multiple of $\tau_{j+1}$.
These are obtained by rounding the unit vectors with $r_i$ coordinate magnitudes in $[\tau_i,\tau_{i-1})$ (for $2\le i\le j+1$) and $r_1$ coordinate magnitudes in $[\tau_1,1]$.

We have nets of this form for $j = 1, 2, \dots, s-1$, where $s = \lceil \frac{1}{\beta} \rceil$. We now have the following claim.
\begin{claim}\label{claim:net-nj}
Fix any $1 \le j < s$. We have
\[ \Pr \left[ \exists \alpha \in \calN_j, \norm{\calM \alpha} < \frac{\theta^p\tau_j^{2p}}{2^{p-1}\lambda} \right] < \exp\left(-\tfrac{1}{4} c m^{j\beta}\right). \]
\end{claim}
\begin{proof}
First consider any single $\alpha \in \calN_j$ and a tuple defining it. By assumption, it has at least $R_j$ coordinates with magnitude $\ge \tau_j$. Also $\norm{\alpha}\le2$, since rounding changes the source unit vector by at most $m\tau_{j+1}\le1$. Thus, from the CAA property applied to $\alpha/\norm{\alpha}$,
\begin{equation}\label{eq:net-anticonc-bound} \Pr \left[ \norm{\calM \alpha} < \frac{\tau_j^p h}{2^{p-1}\lambda} \right] \le \exp \left( -c R_j m^{\beta} \log(1/h) \right).
\end{equation}
The number of vectors in $\calN_j$ for a given tuple of $r_j$ values is clearly bounded by
\[
\binom{m}{R_{j+1}}\big(2\lceil1/\tau_{j+1}\rceil\big)^{R_{j+1}}\le(4m/\tau_{j+1})^{R_{j+1}}.
\]
This counts the non-zero positions, signs, and grid values. We choose $h=(\theta\tau_j)^p$, with $\theta\le1/(8m)$, and argue that as long as these are true,
\begin{equation}
R_{j+1} \log \frac{4m}{\tau_{j+1}} \le \frac{c}{2} R_j m^{\beta} \log(1/h). \label{eq:net-to-show}
\end{equation}
We can simplify this by noting that from our assumptions on $R_j$, $R_{j+1}\le m^\beta R_j$. Thus to show~\eqref{eq:net-to-show}, it suffices to have
\[ \frac{c}{2} \log \frac{1}{h} \ge \log \frac{4m}{\tau_{j+1}}. \]
Since $\theta\le1/(8m)$ and $\tau_{j+1}=\theta^{p+1}\tau_j^{2p}$, we have $\log(4m/\tau_{j+1})\le3\log(1/h)$. Since $c\ge8/\beta\ge4s$
and $h=(\theta\tau_j)^p$, the above inequality holds, and so inequality~\eqref{eq:net-to-show} also holds.

Next, we observe that for any tuple of $\{r_i\}$ values in $\calN_j$, since $r_1 \ge 1$, we have $R_j \ge m^{(j-1)\beta}$, which implies that $R_j m^{\beta} \ge m^{j \beta}$. Using this fact, together with~\eqref{eq:net-to-show}, we first take a union bound over all $\alpha \in \calN_j$ corresponding to a given tuple $(r_i)_{1\le i \le j+1}$ (call this set $\calN_j'$ for now), and obtain
\[ \Pr \left[ \exists \alpha \in \calN_j' : \norm{\calM \alpha} < \frac{\theta^p\tau_j^{2p}}{2^{p-1}\lambda} \right] < \exp \left( -\frac{c}{2} m^{j\beta} \log(1/h) \right). \]
Now since the total number of tuples is easily bounded by $m^{j+1} \le m^s$, taking a further union bound over these choices and simplifying, we obtain the claim. 
\end{proof}

Finally, we have a bigger net for all ``dense'' unit vectors $\alpha$, that have at least $m^{1-\beta}$ coordinates of magnitude $\ge \tau_s$. This net consists of vectors $\in\R^m$ for which (a) every coordinate is an integer multiple of $\tau_{s+1}$, of magnitude at most $\left\lceil\frac1{\tau_{s+1}}\right\rceil\tau_{s+1} \approx 1$, and (b) at least $m^{1-\beta}$ coordinates have magnitude $\ge\tau_s$. We also require these net vectors to have norm at most $2$. Call this net $\calN_0$.

An easy upper bound for the size is
\[
|\calN_0|\le\big(1+2\lceil1/\tau_{s+1}\rceil\big)^m\le\left(\frac{4m}{\tau_{s+1}}\right)^m.
\]
Using this, we have the following:
\begin{claim}\label{claim:net-n0}
\[ \Pr \left[ \exists \alpha \in \calN_0 : \norm{\calM \alpha} < \frac{\theta^p\tau_s^{2p}}{2^{p-1}\lambda}   \right] < \exp \left( -\frac{c}{4} m \right). \]
\end{claim}
\begin{proof}
First, consider any fixed $\alpha \in \calN_0$. Using part (b) of the definition of $\calN_0$, we can use the CAA property after normalizing $\alpha$ as above to obtain
\[ \Pr \left[ \norm{\calM \alpha} < \frac{\tau_s^p h}{2^{p-1}\lambda}   \right] \le \exp \left( -c m \log (1/h) \right),  \]
for any parameter $h\in(0,1)$. As before, we show that this is small enough to take a union bound over $|\calN_0|$ terms. Specifically, we argue that setting $h=(\theta\tau_s)^p$,
\[ \log |\calN_0| \le m \log \frac{4m}{\tau_{s+1}} {}\le3m\log\frac1h \le \frac12 cm \log \frac{1}{h}. \]
The latter holds because $\theta\le1/(8m)$, and $c\ge8/\beta\ge8$.
 Taking a union bound, the failure probability is at most $\exp(-(c/2)m\log(1/h))<\exp(-(c/4)m)$. This completes the proof.
\end{proof}

We can now complete the proof of Theorem~\ref{thm:caa-to-sigma} as follows.

\begin{proof}[Proof of Theorem~\ref{thm:caa-to-sigma}.]
Consider any $\alpha \in \R^m$ with $\norm{\alpha}=1$. Also, we restrict the failure events below to the event that $\norm{\calM} \le L$ (which happens with probability $1-\eta$). Let $s = \lceil \frac{1}{\beta}\rceil$ as before. Now define $r_1, r_2, \dots r_s$ to be the number of entries of $\alpha$ with magnitude in the intervals $[\frac{1}{m}, 1], [\tau_2, \frac{1}{m})$, and so on. We consider two cases:

\emph{Case 1.} $(r_1+r_2+\dots+r_s)\ge m^{1-\beta}$. I.e., there are sufficiently many ``large'' coordinates in $\alpha$.

In this case, we observe that there exists a vector $\alpha'\in\calN_0$ such that $\norm{\alpha-\alpha'}\le m\tau_{s+1}$. Now, we have from Claim~\ref{claim:net-n0} that with high probability, $\norm{\calM \alpha'} \ge \frac{\theta^p\tau_s^{2p}}{2^{p-1}\lambda}$. If this holds, then $\norm{\calM \alpha} \ge \frac{\theta^p\tau_s^{2p}}{2^{p-1}\lambda} - L m\tau_{s+1}$, where $L$ is the spectral norm bound on $\calM$. We choose $\theta$ such that
\[
Lm\tau_{j+1}=Lm\theta^{p+1}\tau_j^{2p}
\le\frac{\theta^p\tau_j^{2p}}{2^p\lambda}
\quad\text{by taking}\quad\theta=(2^{p+2}mL\lambda)^{-1}.
\]
Thus for this value of $\theta$, 
\begin{equation}\label{eq:net-dense-vector-bound} \Pr \left[ \exists \alpha, \norm{\alpha}=1, \text{ satisfying (Case 1)} : ~\norm{\calM \alpha} < \frac{\theta^p\tau_s^{2p}}{2^p\lambda} \right] < \exp\left(-\frac{c}{4} m\right).
\end{equation}

\emph{Case 2.} $(r_1+r_2+\dots+r_s)<m^{1-\beta}$. In this case, we claim that we must have some index $j < s$ such that $R_{j+1}/R_j \le m^{\beta}$. This is because $r_1 \ge 1$ (a unit vector must have some entry of magnitude $> 1/m$), and if the above does not hold, then $R_j\ge m^{(j-1)\beta}$. Plugging $j = s$ gives a contradiction to our assumption that $(r_1 + r_2 + \dots r_s) < m^{1-\beta}$.

Thus, let $j$ be the smallest index for which we have $R_{j+1}/R_j \le m^{\beta}$. We can now consider the net $\calN_j$ and find some $\alpha' \in \calN_j$ such that $\norm{\alpha - \alpha'} < m\tau_{j+1}$.

We can use an argument identical to the one in (Case 1), this time leveraging Claim~\ref{claim:net-nj}, to conclude that for all $j$,
\[
\begin{aligned}
&\Pr \left[ \exists \alpha, \norm{\alpha}=1, \text{satisfying (Case 2)}: \norm{\calM \alpha} < \frac{\theta^p\tau_s^{2p}}{2^p\lambda} \right]\\
&\qquad<\sum_{1\le j<s} \exp\left( -\frac{c}{4} m^{j\beta} \right).
\end{aligned}
\]
The total failure probability from the two cases is at most $s e^{-cm^\beta/4}\le e^{-cm^\beta/8}$, since $s\le c/4$ and $\log s\le c/8$, and plugging in the chosen values of $\theta, s$ and adding $\eta$, this completes the proof.
\end{proof}

\section{Higher Order Lifts of Smoothed Matrices}\label{sec:contraction}\label{sec:kronecker}
\newcommand{\flatW}{W^{\textrm{flat}}}
We provide the following theorem. 

\begin{theorem}\label{thm:kronecker:sym}
Suppose $d \in \N$, and let $\Phi: \Sym^d(\R^n) \to \R^D$ be an orthogonal projection of rank $R = \delta \binom{n+d-1}{d}$ for some constant $\delta>0$, and let $\Sym_d: (\R^{n})^{\otimes d} \to \Sym^d(\R^n)$ be the orthogonal projection on to the symmetric subspace of $(\R^n)^{\otimes d}$. Let $U = (u_i : i \in [m]) \in \R^{n \times m}$ be an arbitrary matrix, and let $\tilde{U}$ be a random $\rho$-perturbation of $U$. Then there exists a constant $c_d>0$ such that for $m \le c_d \delta n$, with probability at least $1-\exp\big(- \Omega_{d,\delta} (n)\big)$, we have the least singular value
\begin{align}
\sigma_{\binom{m+d-1}{d}}(\Phi \tilde{U}^{\veep d})  &\ge \Omega_d\left(\frac{\rho^d}{m^{d/2}n^{d/2}}\right) = \Omega_d\left(\frac{\rho^d}{n^{d}}\right), \text{ where }\nonumber \\
\tilde{U}^{\veep d} &\coloneqq \Big(\Sym_d(\tilde{u}_{i_1} \otimes\tilde{u}_{i_2} \dots \otimes \tilde{u}_{i_d}): 1 \le i_1 \le i_2 \le \dots \le i_d \le m\Big).
\end{align}
\end{theorem}

In the above statement, one can also consider an arbitrary linear operator $\Phi$ and suffer an extra factor of $\sigma_R(\Phi)$ in the least singular value bound (by considering the projector onto the span of the top $R$ singular vectors). In the rest of the section, we assume that $\Phi$ is an orthogonal projector of rank $R$ without loss of generality.

\begin{figure}[t]
\centering
\begingroup

\definecolor{tensorblue}{HTML}{315EA5}
\definecolor{tensorgreen}{HTML}{527D35}
\definecolor{fiberpurple}{HTML}{A575C5}
\definecolor{fibergold}{HTML}{F6BD25}
\definecolor{fiberyellow}{HTML}{F0E843}
\definecolor{selectedtriangle}{HTML}{C43D35}

\resizebox{\linewidth}{!}{
\begin{tikzpicture}[
    x=1cm,y=1cm,
    line join=round,line cap=round,
    dim/.style={latex-latex,densely dashed,thin,black!75},
    edge/.style={line width=.65pt}
]

\newcommand{\tensorbody}[6]{
    \fill[white] (0,0) rectangle (#2,#3);
    \fill[black!9]
        (#2,0)--({#2+#4},#5)--
        ({#2+#4},{#3+#5})--(#2,#3)--cycle;
    \fill[white]
        (0,#3)--(#2,#3)--
        ({#2+#4},{#3+#5})--(#4,{#3+#5})--cycle;
    \foreach \a in {0,...,#1}{
        \foreach \b in {0,...,5}{
            \draw[#6!8,thin]
                ({#2*\a/#1},{#3*\b/6})--
                ({#2*\a/#1+#4/#1},{#3*\b/6+#5/#1});
        }
    }
}

\newcommand{\tensorgrid}[6]{
    \foreach \a in {0,...,#1}{
        \draw[#6,edge]
            ({#2*\a/#1},0)--
            ({#2*\a/#1},#3)--
            ({#2*\a/#1+#4},{#3+#5});
        \draw[#6,edge]
            ({#4*\a/#1},{#3+#5*\a/#1})--
            ({#2+#4*\a/#1},{#3+#5*\a/#1})--
            ({#2+#4*\a/#1},{#5*\a/#1});
    }
    \foreach \b in {0,...,6}{
        \draw[#6,edge]
            (0,{#3*\b/6})--(#2,{#3*\b/6})--
            ({#2+#4},{#3*\b/6+#5});
    }
    \draw[#6,edge]
        (0,0)--(0,#3)--(#4,{#3+#5})--
        ({#2+#4},{#3+#5})--({#2+#4},#5)--(#2,0)--cycle;
}

\newcommand{\columnvector}[3]{
    \begin{scope}[shift={(#1,#2)}]
        \fill[#3!85] (0,0) rectangle (.60,4.2);
        \fill[#3!70!black]
            (.60,0)--(.83,.23)--(.83,4.43)--(.60,4.2)--cycle;
        \fill[#3!60]
            (0,4.2)--(.60,4.2)--(.83,4.43)--(.23,4.43)--cycle;
        \foreach \b in {0,...,6}{
            \draw[tensorblue,edge]
                (0,{.7*\b})--(.60,{.7*\b})--(.83,{.7*\b+.23});
        }
        \draw[tensorblue,edge]
            (0,0)--(0,4.2)--(.23,4.43)--(.83,4.43)--
            (.83,.23)--(.60,0)--cycle;
        \draw[tensorblue,edge]
            (.60,0)--(.60,4.2)--(.83,4.43);
    \end{scope}
}

\begin{scope}[shift={(2.05,1.40)}]
    \tensorbody{6}{4.5}{4.2}{1.35}{1.35}{tensorblue}

    \fill[fiberpurple!83] (3,0) rectangle (3.75,4.2);
    \fill[fiberpurple!65]
        (3,4.2)--(3.75,4.2)--(3.975,4.425)--(3.225,4.425)--cycle;
    \fill[fibergold!85]
        (5.625,1.125)--(5.85,1.35)--
        (5.85,5.55)--(5.625,5.325)--cycle;
    \fill[fibergold!75]
        (4.875,5.325)--(5.625,5.325)--
        (5.85,5.55)--(5.10,5.55)--cycle;
    \fill[fiberyellow]
        (2.175,4.875)--(2.925,4.875)--
        (3.15,5.1)--(2.4,5.1)--cycle;

    \tensorgrid{6}{4.5}{4.2}{1.35}{1.35}{tensorblue}

    \draw[dim] (-.26,0)--node[left=3pt] {$R$}(-.26,4.2);
    \draw[dim] (-.17,4.34)--
        node[above left=-1pt] {$n$}(1.18,5.69);
    \draw[dim] (1.35,5.76)--
        node[above=1pt] {$n$}(5.85,5.76);
\end{scope}

\draw[tensorgreen,line width=1.15pt,fill=white]
    (3.52,7.56)--(8.02,7.56)--(8.64,8.34)--(4.14,8.34)--cycle;
\node[font=\Large] at (6.08,7.94) {$\widetilde U$};
\draw[dim] (3.37,7.58)--
    node[above left=-1pt] {$m$}(3.99,8.36);

\draw[tensorgreen,line width=1.15pt,fill=white]
    (-.12,5.77)--(1.23,7.12)--(2.34,6.64)--(.99,5.29)--cycle;
\node[font=\Large,rotate=45] at (1.11,6.20) {$\widetilde U$};
\draw[dim] (-.23,5.93)--
    node[above left=-1pt] {$n$}(1.12,7.28);
\draw[dim] (-.10,5.54)--node[below=1pt] {$m$}(1.01,5.06);

\columnvector{8.15}{1.35}{fiberpurple}
\columnvector{8.95}{2.20}{fibergold}
\columnvector{9.75}{1.75}{fiberyellow}

\node at (5.60,.79)
    {$\Phi\in\mathbb R^{R\times n\times n}$
     comprised of $n^2$ vectors of dim $R$};

\draw[-latex,line width=.85pt,black!65]
    (10.92,4.05)--(12.13,4.05);

\begin{scope}[shift={(13.00,1.40)}]
    \tensorbody{4}{3.6}{4.2}{1.24}{1.24}{tensorgreen}

    \foreach \i in {0,...,3}{
        \foreach \j in {0,...,3}{
            \ifnum\j<\i\else
                \fill[selectedtriangle!32]
                    ({.9*\j+.31*\i},{4.2+.31*\i})--
                    ({.9*(\j+1)+.31*\i},{4.2+.31*\i})--
                    ({.9*(\j+1)+.31*(\i+1)},{4.2+.31*(\i+1)})--
                    ({.9*\j+.31*(\i+1)},{4.2+.31*(\i+1)})--cycle;
            \fi
        }
    }

    \tensorgrid{4}{3.6}{4.2}{1.24}{1.24}{tensorgreen}

    \draw[selectedtriangle,line width=1.05pt]
        (.31,4.51)--(1.21,4.51)--(1.52,4.82)--(2.42,4.82)--
        (2.73,5.13)--(3.63,5.13)--(3.94,5.44);

    \draw[dim] (-.26,0)--node[left=3pt] {$R$}(-.26,4.2);
    \draw[dim] (-.17,4.34)--
        node[above left=-1pt] {$m$}(1.07,5.58);
    \draw[dim] (1.24,5.65)--
        node[above=1pt] {$m$}(4.84,5.65);
\end{scope}

\node at (15.48,.79) {$\mathcal M$};
\node[font=\small] at (15.48,.23)
    {has $\binom{m+1}{2}$ symmetric vectors of dim $R$};

\end{tikzpicture}
}
\endgroup
\caption{\textit{Left}: The linear operator $\Phi: \Sym(\R^{n \times n}) \to \R^R$ interpreted as a tensor consisting of a $n \times n$ array of $R$-dimensional vectors. There are {\em smoothed} or random contractions applied using the matrix $\tilde{U}\in \R^{n \times m}$. \textit{Right}: The operator $\Phi(\tilde{U} \veep \tilde{U}): \Sym(\R^{m \times m}) \to \R^R$ interpreted as an $\binom{m+1}{2}$ array of $R$-dimensional vectors. The shaded upper triangle indexes the $\binom{m+1}{2}$ columns of $\mathcal M$. Theorem~\ref{thm:kronecker:sym} shows that under the conditions of the theorem, with high probability the robust rank of this operator is $\binom{m+1}{2}$ i.e, the least singular value of the matrix is inverse polynomial.}
\label{fig:symmetric-contraction}
\end{figure}

 Applying Theorem~\ref{thm:kronecker:sym} along with the block leave-one-out approach (see Lemma~\ref{lemma:BlockLeaveoneOut}) we arrive at the following corollary.

\begin{corollary}\label{corr:kron:blocks}
Suppose $d,t \in \N$ and let $1 \geq \delta_1>\delta_2 >0$ be given.  Also let $\Phi: \Sym^d(\R^n) \to \R^D$ be an orthogonal projection of rank $R \geq \delta_1 \binom{n+d-1}{d}$. Let $\{U_j\}_{j=1}^t \subset \R^{n \times m}$ be an arbitrary collection of $n \times m$ matrices, and  for each $j$, let $\tilde{U}_j$ be a random $\rho$-perturbation of $U_j$. Then there exists a constant $c_d >0$ such that if $t\binom{m+d-1}{d} \leq \delta_2 \binom{n+d-1}{d}$ and $m \leq c_d (\delta_1-\delta_2) n$, then with probability at least $1-\exp\big(- \Omega_{d,\delta_1,\delta_2} (n)\big)$, we have the least singular value
\begin{align}
\label{eq:BlockKronLowerBound}
\sigma_{t\binom{m+d-1}{d}}\left(\Phi \begin{bmatrix} \tilde{U}_1^{\veep d} & \tilde{U}_2^{\veep d} & \dots & \tilde{U}_t^{\veep d} \end{bmatrix}\right)  &\ge \frac{\rho^d}{\sqrt{t} n^{O(d)}}.
\end{align}

\end{corollary}

\begin{proof}
For each $j=1,\dots,t$, let $\Pi_{-j}$ denote the orthogonal projection onto
\[
\mathrm{range}(\Phi) \cap \Big(\mathrm{range} \big([\tilde{U}_1^{\veep d} \ \dots \tilde{U}_{j-1}^{\veep d}\ \tilde{U}_{j+1}^{\veep d}\  \dots \tilde{U}_t^{\veep d}]\big)\Big)^\perp.
\]
We first lower bound the least singular value of $\Pi_{-j} (\tilde{U}_j^{\veep d}).$ Observe that from our assumptions, the rank of $\Pi_{-j}$ is at least
\[
\delta_1 \binom{n+d-1}{d}+(1-\delta_2) \binom{n+d-1}{d} - \binom{n+d-1}{d} = (\delta_1-\delta_2) \binom{n+d-1}{d}.
\]
Taking $c_d$ to be the constant from Theorem \ref{thm:kronecker:sym}, we can then apply Theorem \ref{thm:kronecker:sym} to conclude that with probability at least $1-\exp(-\Omega_{d,\delta_1,\delta_2}(n))$ we have
\[
\sigma_{\binom{m+d-1}{d}} (\Pi_{-j} \tilde{U}_j^{\veep d}) \geq \frac{\rho^d}{n^{O(d)}}.
\]
The result now follows by applying the block leave one out bound of Lemma \ref{lemma:BlockLeaveoneOut}.
\end{proof}

\paragraph{Remark.} In a previous version of this paper (\url{https://arxiv.org/abs/2405.01517v2}), the proof of Theorem~\ref{thm:kronecker:sym} had an error, specifically in Lemma 5.5 and in the decoupling step in reducing the symmetric Kronecker to the asymmetric Kronecker product (for $d>2$). The following corrected proof bypasses both steps, is shorter, and conceptually simpler. 

\subsection{Proof of Theorem \ref{thm:kronecker:sym}}
\label{sec:proof_of_kronecker_sym}
We will need the following useful lemma about anti-concentration of polynomial maps. 

\begin{lemma}[{\cite[Theorem~2.3]{BCPV}}]\label{lemma:bcpv_operator_applied_to_rank_ones}
    Let $W \subseteq \Sym^d(\R^n)$ be a subspace of dimension at least $\delta n^d$, for some $\delta \in (0,1)$. Then, if $\tu$ is a $\rho$-perturbed vector in $\R^n$, we have, with probability at least $1-\exp(-\Omega_d(n))$, 
    \begin{equation*}
        \norm{\Pi_W(\tu^{\otimes d})} \ge \Omega\inparen{\frac{\rho^d}{n^{d/2}}}
    \end{equation*}
\end{lemma}

As described in the overview, we will use a leave-one-out lemma applied to an arbitrary set of spanning vectors. To this end, we use the following definition. 

\begin{definition}[Frames]
    Let $\cV$ be a vector space. For $\tau > 0$, a $\tau$-frame is defined as a finite set of unit vectors $S \subset \cV$ such that 
    \begin{equation*}
        \inparen{\sum_{v \in S} \left|\Iprod{v, x}\right|^2}^{1/2} \ge \tau \norm{x}, 
        \qquad \textnormal{for every } x \in \cV. 
    \end{equation*}
\end{definition}

Equivalently, identify $\cV$ with $\R^{\dim \cV}$ by choosing an orthonormal basis and let $B_S \in \R^{|S| \times \dim \cV}$ be the matrix whose rows are the vectors in $S$. Then, $S$ is a $\tau$-frame if and only if the least singular value of $B_S$ is at least $\tau$. 

We now prove a general leave-one-out statement, that uses an arbitrary frame.

\begin{lemma}\label{lem:leave_one_out_arbitrary_basis}
    Let $T: \cV \to \cW$ be a linear map between vectors spaces $\cV$ and $\cW$. Let $S$ be a $\tau$-frame for $\cV$. 
    Then, the least singular value of $T$ satisfies 
    \begin{equation*}
        \sigma_{\min{}}(T) \ge \frac{\tau}{\sqrt{|S|}} \inf_{v \in S} \norm{\Pi_{T\left(\inangle{v}^\perp\right)}^\perp T(v)}.
    \end{equation*} where $\inangle{v}^\perp$ is the subspace of vectors in $\cV$ orthogonal to $v$.
\end{lemma}

Note that if we use the standard basis as our frame in Lemma (\ref{lem:leave_one_out_arbitrary_basis}), we recover the standard column leave-one-out lemma i.e. Lemma (\ref{lem:leaveoneout}). 

\begin{proof}[Proof of Lemma (\ref{lem:leave_one_out_arbitrary_basis})]
    Let $S = \{v_1, \ldots, v_r\}$, where $r = |S|$. 
    Let $x^* \in \cV$ be the unit vector corresponding to the least singular value of $T$. Then, $\sigma_{\min_{}}(T) = \norm{T(x^*)}$,
    and we have that 
    \begin{align*}
        \sum_{i \in [r]} \left|\Iprod{x^*, v_i}\right|^2 \ge \tau^2.
    \end{align*} since $\{v_1, \ldots, v_r\}$ is a $\tau$-frame.
    Hence, there exists an $i^* \in [r]$ for which $\left|\Iprod{x^*, v_{i^*}}\right| \ge \tau/\sqrt{r}$. Further, we can write 
    \begin{equation*}
        x^* = \Iprod{x^*, v_{i^*}} v_{i^*} + x_{\perp}^*,
    \end{equation*} where $x^*_\perp \perp v_{i^*}$.
    Finally, we get 
    \begin{align*}
        \sigma_{\min_{}}(T) &= \norm{T(x^*)} \ge \norm{\Pi_{T\left(\inangle{v_{i^*}}^\perp\right)}^\perp T(x^*)} \ge \left|\Iprod{x^*, v_{i^*}}\right| \cdot \norm{\Pi_{T\left(\inangle{v_{i^*}}^\perp\right)}^\perp T(v_{i^*})} \\
        &\ge \frac{\tau}{\sqrt{r}} \inf_{i \in [r]} \norm{\Pi_{T\left(\inangle{v_{i}}^\perp\right)}^\perp T(v_i)}.
    \end{align*} 
\end{proof}

The matrix in Theorem \ref{thm:kronecker:sym} is a linear operator acting on the space of symmetric $d$-th order tensors, and our proof  will involve considering its action on a frame of symmetric rank-one tensors. We next show that there exists a good frame consisting only of symmetric rank one tensors for any $d$.

\begin{lemma}\label{lem:good_basis_of_rank_ones}
    Let $d\ge 1$ be an integer, and $m \in \N$. There exist vectors $v_1, \ldots, v_r \in \R^m$ such that $\{v_1^{\otimes d}, \ldots, v_r^{\otimes d}\}$ is a $\tau_d$-frame of $\Sym^d(\R^m)$, where $\tau_d =  \Omega_d(1)$ and $r \le 2^{d-1}\cdot\binom{m+d-1}{d}$. 
\end{lemma}

\begin{proof}[Proof of Lemma (\ref{lem:good_basis_of_rank_ones})]
    The case $d=1$ is easy - consider an orthonormal basis for $\R^m$. Let us focus on $d \ge 2$.

    We provide a deterministic construction that achieves the required lower bound. Let $e_1,\ldots, e_m$ be an orthonormal basis for $\R^m$, and consider the set of vectors given by lifts of signed linear combinations:
    \begin{equation*}
        S = \inbrace{\frac{\inparen{e_{i_1} + \varepsilon_2 e_{i_2} + \cdots + \varepsilon_de_{i_d}}^{\otimes d}}{\norm{e_{i_1} + \varepsilon_2 e_{i_2} + \cdots + \varepsilon_de_{i_d}}^d} : 1 \le i_1 \le i_2 \le \cdots \le i_d \le m, \varepsilon_2, \ldots, \varepsilon_d \in \{-1,+1\}}. 
    \end{equation*}
    We will now show that $S$ is a $\frac{\sqrt{2^{d-1} d!}}{d^d}$-frame of $\Sym^d(\R^m)$. For any $H \in \Sym^d(\R^m)$, we have 
    \begin{equation}
    \label{eqn:frame_construction_inner_prod}
        \begin{split}
        \sum_{\substack{1\le i_1 \le \cdots \le i_d \le m \\ \varepsilon_2,\ldots, \varepsilon_d \in \{-1,+1\}}} &\Iprod{\frac{\inparen{e_{i_1} + \varepsilon_2 e_{i_2} + \cdots + \varepsilon_de_{i_d}}^{\otimes d}}{\norm{e_{i_1} + \varepsilon_2 e_{i_2} + \cdots + \varepsilon_de_{i_d}}^d}, H}^2 \\
        &\ge 
        \frac{2^{d-1}}{d^{2d}}\sum_{1 \le i_1 \le \cdots \le i_d \le m} \Esymb_\varepsilon \Iprod{\inparen{e_{i_1} + \varepsilon_2 e_{i_2} + \cdots + \varepsilon_de_{i_d}}^{\otimes d}, H}^2
        \end{split}
    \end{equation}
    where the expectation is taken over $\varepsilon_2,\ldots, \varepsilon_d$ drawn uniformly at random from $\{-1,+1\}$. For any valid $i_1 \le i_2 \le \cdots \le i_d$, since $H$ is symmetric, we can expand the inner product as 
    \begin{align*}
        \Iprod{\inparen{e_{i_1}+\varepsilon_2e_{i_2} + \cdots + \varepsilon_d e_{i_d}}^{\otimes d}, H} = d!\cdot (\varepsilon_2\cdots\varepsilon_d)\cdot H_{i_1,\ldots, i_d} + \Iprod{H, T_\varepsilon^{i_1,\ldots,i_d}}
    \end{align*} where $T_\varepsilon^{i_1,\ldots, i_d}$ consists of all the terms in the expansion of the tensor product which do not use each of $e_{i_1},\ldots, e_{i_d}$ exactly once. In particular, this means that $\Esymb_{\varepsilon}(\varepsilon_2\cdots\varepsilon_d)T_{\varepsilon}^{i_1,\ldots,i_d}$ = 0, and we have 
    \begin{equation}\label{eqn:epsilon_expectation_simplified}
        \Esymb_{\varepsilon}\left[\Iprod{\inparen{e_{i_1}+\varepsilon_2e_{i_2} + \cdots + \varepsilon_d e_{i_d}}^{\otimes d}, H}^2\right] = (d!)^2 H_{i_1,\ldots, i_d}^2 +\Esymb_{\varepsilon}\Iprod{H,T_\varepsilon^{i_1,\ldots, i_d}}^2 \ge (d!)^2 H_{i_1,\ldots, i_d}^2. 
    \end{equation}
    Substituting equation (\ref{eqn:epsilon_expectation_simplified}) into equation (\ref{eqn:frame_construction_inner_prod}), 
    \begin{equation*}
        \begin{split}
                \sum_{\substack{1\le i_1 \le \cdots \le i_d \le m \\ \varepsilon_2,\ldots, \varepsilon_d \in \{-1,+1\}}} &\Iprod{\frac{\inparen{e_{i_1} + \varepsilon_2 e_{i_2} + \cdots + \varepsilon_de_{i_d}}^{\otimes d}}{\norm{e_{i_1} + \varepsilon_2 e_{i_2} + \cdots + \varepsilon_de_{i_d}}^d}, H}^2 \\
                &\ge \frac{2^{d-1}(d!)^2}{d^{2d}} \sum_{1 \le i_1 \le \cdots \le i_d \le m} H_{i_1,\ldots, i_d}^2 \\
                &\ge \frac{2^{d-1}d!}{d^{2d}}\sum_{i_1,\ldots, i_d = 1}^m H_{i_1,\ldots, i_d}^2 \\ 
                &= \frac{2^{d-1}d!}{d^{2d}} \norm{H}_F^2
        \end{split}
    \end{equation*}
    which gives us the desired result. 
\end{proof}

We are now ready to prove Theorem \ref{thm:kronecker:sym}. 

\begin{proof}[Proof of Theorem \ref{thm:kronecker:sym}]
    Let $\calM = \Phi(\tU^{\veep d})$; we view it as a linear map $\calM:\Sym_d(\R^m) \to \R^D$. 
    By Lemma (\ref{lem:good_basis_of_rank_ones}), there exists a set of unit vectors $S = \{v_1,\ldots, v_r\} \subset \R^m$ of size at most $2^{d-1}\binom{m+d-1}{d}$, such that $\{v_1^{\otimes d},\ldots, v^{\otimes d}_r\}$ is a $\tau_d$-frame, for $\tau_d$ a constant depending only on $d$. Using this frame with Lemma (\ref{lem:leave_one_out_arbitrary_basis}), 
    \begin{equation}\label{eqn:least_sing_val_tau}
        \sigma_{\binom{m+d-1}{d}}(\calM) \ge \frac{\tau_d}{2^{(d-1)/2} \sqrt{\binom{m+d-1}{d}}} \inf_{v \in S} \norm{\Pi_{\calM\inparen{\inangle{v^{\otimes d}}^\perp}}^\perp \calM(v^{\otimes d})}.
    \end{equation}
    We will show that, for any fixed $v$, projecting onto the subspace orthogonal to $\calM\left(\inangle{v^{\otimes d}}^\perp\right)$ leaves enough length in $\calM(v^{\otimes d})$ with high probability. We then union bound over all vectors in $S$.

    First, consider a fixed $v$. We can write $\calM(v^{\otimes d}) = \Phi(\tU v)^{\otimes d} = \Phi(\tz)^{\otimes d}$ where $\tz = \tU v$. Note that $\tz$ is a $\rho$-perturbed random vector in $\R^n$. Let $\tcT_v = \calM\left(\inangle{v^{\otimes d}}^\perp\right)$. Then, we want to show that $\norm{\Pi_{\tcT_v}^\perp \Phi \tz^{\otimes d}}$ is inverse polynomially large with high probability. 
    
    Now, the subspace in $\Sym^d(\R^m)$ that is orthogonal to $v^{\otimes d}$ is given by $\Sym_d\left(\inangle{v}^\perp \otimes \R^m \otimes \cdots \otimes \R^m \right)$, and so the subspace we are projecting out can be characterized as 
    \begin{equation*}
        \tcT_v = \Phi \cdot {\Sym_d\inparen{\tU \cdot \inangle{v}^\perp \otimes \tcU \otimes \cdots \otimes \tcU}}
    \end{equation*} where $\tcU = \colspn(\tU)$. Note that the vector $\tz = \tU v$ is independent of any vector in the subspace $\tU \cdot \inangle{v}^\perp$; the subspace $\tcT_v$ and the vector $(\tz)^{\otimes d}$ are dependent only in the modes containing $\tcU$. We therefore project out a larger subspace $\tcS_v\supseteq\tcT_v$ that is independent of $\tz$, given by 
    \begin{equation*}
        \tcS_v = \Phi \cdot {\Sym_d\inparen{\tU \cdot \inangle{v}^\perp \otimes \R^n \otimes \cdots \otimes \R^n}}. 
    \end{equation*}
    As described above, $\tcS_v$ is independent of $\tz$. Further, we have that $\dim \tcS_v \le n^{d-1}(m-1)$ and we can choose $m$ to be small enough so that $\dim {\tcS_v} \le \frac{1}{4}\rank (\Phi) \implies \rank(\Pi_{\tcS_v}^\perp \Phi) \ge \frac{3}{4} \rank(\Phi) = \frac{3\delta}{4}\binom{n+d-1}{d}$. Choosing $m$ to be small enough, by Lemma (\ref{lemma:bcpv_operator_applied_to_rank_ones}), we have 
    \begin{equation*}
        \norm{\Pi_{\tcT_v}^\perp \Phi \tz^{\otimes d}} \ge \norm{\Pi_{\tcS_v}^\perp \Phi \tz^{\otimes d}} \ge \Omega\inparen{\frac{\rho^d}{n^{d/2}}}. 
    \end{equation*}
    with probability at least $1 - \exp(-\Omega_{\delta, d}(n))$.
    
    Finally, we do a union bound over all $v \in S$ in equation (\ref{eqn:least_sing_val_tau}), and get that, with probability $1 - \exp(-\Omega_{\delta, d}(n))$, 
    \begin{equation*}
        \sigma_{\binom{m+d-1}{d}}(\calM) \ge C_{d,\delta}\inparen{\frac{\rho^d}{m^{d/2}n^{d/2}}}
    \end{equation*}
    where $C_{d,\delta}$ is a constant depending on $d, \delta$. 
\end{proof}

\section{Anticoncentration of a vector of homogeneous polynomials} \label{sec:jacobian}
\subsection{Overview}
\label{sec:overview:anticonc}

Formally, we consider the following setting: let $p_1, p_2, \dots, p_N$ be a collection of homogeneous polynomials over $n$ variables $(x_1, x_2, \dots, x_n)$, and define 
\begin{align}
P(x)= \begin{bmatrix}
    p_1(x) \\
    p_2(x) \\
    \vdots \\
    p_N(x)
    \end{bmatrix}
\end{align}

Our goal will be to show anticoncentration results for $P$. Specifically, we want to prove that $\Pr[ \norm{P(\xtil) - y} < \eps]$ is small for all $y$, where $\xtil$ is a perturbation of some (arbitrary) vector $x \in \R^n$.

We first observe that such a statement is not hard to prove if we know that the Jacobian $J(x)$ of $P(x)$ has many large singular values at \emph{every} $x$, and if the perturbation $\rho$ is small enough. This is because around the given point $x$, we can consider the linear approximation of $P(\tilde{x})$ given by the Jacobian. Now as long as the perturbation has a high enough projection onto the span of the corresponding singular vectors of $J(x)$, $P(\tilde{x})$ can be shown to have desired anticoncentration properties (by using the standard anticoncentration result for Gaussians). Finally, if $J(x)$ has $k$ large singular values, a random $\rho$-perturbation will have a large enough projection to the span of the singular vectors with probability $1-\exp(-k)$.

Now, in the applications we are interested in, the polynomials $P$ tend to have the Jacobian property above for ``typical'' points $x$, but not all $x$. Our main result here is to show that this property suffices. Specifically, suppose we know that for every $x$, the Jacobian at a $r$-perturbed point has $k$ singular values of magnitude $\ge \tau$ with high probability for each $r\in[\rho/2,\rho]$. Then, in order to show anticoncentration, we view the $\rho$ perturbation of $x$ as occurring in two independent steps: first perturb by $\rho \sqrt{1-z^2}$ for some parameter $z$, and then perturb by $\rho z$. The key observation is that for Gaussian perturbations, this is identical to a $\rho$ perturbation!

This gives an approach for proving anticoncentration. We use the fact that the first perturbation yields a point with sufficiently many large Jacobian singular values with high probability, and combine this with our earlier result (discussed above) to show that if $z$ is small enough, the linear approximation can indeed be used for the second perturbation, and this yields the desired anticoncentration bound. 

\emph{Applications.}  The simplest application for our framework is the setting where $\calM$ has columns being $\tilde{u}_i \otimes \tilde{v}_i$, for some $\rho$-perturbations of underlying vectors $u_i, v_i$. (This setting was studied in~\cite{BCMV, ADMPSV18} and already had applications to parameter recovery in statistical models.) Here, we can show that $\calM$ has the CAA property. To show this, we consider some combination $\sum_i \alpha_i (\tilde{u}_i \otimes \tilde{v}_i)$ with $k$ ``large'' coefficients in $\alpha$, and show that in this case, the Jacobian property holds. Specifically, we show that the Jacobian has  $\Omega(n\min\{k,n\})$ large singular values. This establishes the CAA property, which in turn implies a lower bound on $\sigma_{\min}(\calM)$, assuming the unperturbed vectors have polynomially bounded norms. This gives an alternative proof of the results of the works above.

\subsection{Jacobian rank property and anticoncentration results}

We give a sufficient condition for proving such a result, in terms of the Jacobian of $P$. (See Section~\ref{sec:prelims} for background.)

\begin{definition}[Jacobian rank property]\label{defn:jacobian-rank}
We say that $P$ has the Jacobian rank property with parameters $(k,\tau,\gamma)$ at a point $x$ with perturbation scale $r>0$ if the matrix $J(\xtil)$ has at least $k$ singular values of magnitude $\ge \tau$, with probability at least $1-\gamma$. Here, $\xtil = x + \eta$, where $\eta \sim \calN(0, r^2 I_n)$ is a perturbation of the vector $x$.
\end{definition}

For many interesting settings of $P$, the Jacobian rank property turns out to be quite simple to prove. Our main result now is that the property above implies an anticoncentration bound for $P$.

\begin{theorem}\label{thm:jacobian-to-antic}
Suppose $P(x)$ defined as above satisfies the Jacobian rank property with parameters $(k,\tau,\gamma)$ at $x$ for every scale $r\in[\rho/2,\rho]$, and suppose further that the
Jacobian $P'$ is $M$-Lipschitz  in operator norm on a convex set $\mathcal B$ containing $x$, with $M>0$.
Let  $\tilde{x}$ be a $\rho$-perturbation of $x$ and set $\xi=\Pr[\tilde{x}\notin\mathcal B]$. Then for any $h>0$, we have

\[
\forall y\in\R^N,\quad
\Pr\!\left[\norm{P(\xtil)-y}<\min\!\left\{\frac{\tau\rho\sqrt h}{16},\frac{\tau^2h}{2048Mnk}\right\}\right]
\le\gamma+h^{k/4}+2\xi.
\]
\end{theorem}

A key ingredient in the proof is the following ``linearization'' based lemma. 

\begin{lemma}\label{lem:linearization}
Suppose $x$ is a point at which the Jacobian $J(x)$ of a polynomial map $P$ has at least $k$ singular values of magnitude $\ge \tau$. Also suppose that  $J$ is $M$-Lipschitz in operator norm on a convex set $\mathcal B$ containing $x$, with $M>0$. Then, for ``small'' perturbations, $0 < \rho < \frac{\tau}{4Mnk}$, we have that for any $\eps >0$,
\[ \forall y, ~\Pr[\norm{ P(\xtil ) - y } < \eps ,\ \xtil\in\mathcal B] < \left( \frac{2\eps}{\tau\rho} \right)^k + \left( \frac{32M \rho n k}{\tau} \right)^{k/2}. \]
For the unconditional probability, add $\Pr[\xtil\notin\mathcal B]$.
\end{lemma}

We remark that the lemma does not imply Theorem~\ref{thm:jacobian-to-antic} directly because it only applies to the case where the perturbation $\rho$ is much smaller than the singular value threshold $\tau$. 

\begin{proof}
Let $\eta$ be the random perturbation of $x$ as in the lemma statement. We have
\[ P(x+\eta) = P(x) + J(x)^T \eta + E(\eta),\]
where $E(\eta)$ is an error term, bounded in magnitude by $M \norm{\eta}^2$ because of our assumption of the Jacobian being Lipschitz, provided $x+\eta\in\mathcal B$ so the segment from $x$ to $x+\eta$ lies in $\mathcal B$. Now, the desired probability is equivalent to
\[ \Pr \left[ J(x)^T \eta + E(\eta) \in \text{Ball}(y - P(x), \eps) ,\ x+\eta\in\mathcal B \right].\]
From the bound on $\eta$, the above probability can be upper bounded by
\[ \Pr \left[ J(x)^T \eta \in \text{Ball}\left( y - P(x), \eps + M \norm{\eta}^2 \right) \right]. \]

Let us denote the event in the parentheses above by $\cE$.
Now, consider the top $k$ singular directions of $J(x)$; suppose the singular values are $\sigma_1, \sigma_2, \dots, \sigma_k$, and suppose $\eta_1, \eta_2, \dots, \eta_k$ are the components of $\eta$ along these directions. By hypothesis, $\sigma_i \ge \tau$ for all $i \le k$. Thus if $\cE$ occurs, we also have,
\begin{equation}\label{eq:small-ball-i} 
\forall i \le k,~ \sigma_i \eta_i \in \text{Ball}\left( (y-P(x))_i, \eps + M \norm{\eta}^2  \right) .
\end{equation}

Let $\theta > 1$ be a parameter that we set later. We note that by Gaussian tail bounds,
\[ \Pr [ \norm{\eta}^2 > 4n \rho^2 k\theta ] \le \exp(- k \theta).  \]

In what follows, restrict the event in~\eqref{eq:small-ball-i} to the event $\norm{\eta}^2 \le 4n \rho^2 k \theta$. Then, the probability in~\eqref{eq:small-ball-i} is upper bounded by
\[ \left( \frac{\eps + 4M \rho^2 n k \theta}{\tau \rho} \right)^k {}+e^{-k\theta} \le \left( \frac{2\eps}{\tau\rho} \right)^k + \left( \frac{8M \rho n k \theta}{\tau} \right)^k {}+e^{-k\theta}. \]

Write $u=M\rho nk/\tau$. If $u\ge1/32$, the desired bound is trivial. Otherwise, we will choose the parameter $\theta = \log \left( \frac{\tau}{M\rho nk} \right)$. This ensures that the term $\exp(-k\theta)$ is roughly the same order of magnitude as the last term on the RHS above. Indeed, the two error terms sum to less than $(32u)^{k/2}$, since $\sqrt u(8\log(1/u)+1)\le (8\log 32+1)/\sqrt{32}<\sqrt{32}$.
\end{proof}

\begin{proof}[Proof of Theorem~\ref{thm:jacobian-to-antic}.]
The main idea in the proof is to view the perturbation $x \rightarrow \tilde{x}$ as occurring in two independent steps $x \rightarrow x' \rightarrow \tilde{x}$, where the first perturbation has scale $\rho \sqrt{1-z^2}$ and the second perturbation has scale $\rho z$. By standard properties of Gaussian perturbations, this is equivalent to a $\rho$ perturbation of $x$. We pick the parameter $z\le1/2$ carefully (later).

Using the Jacobian rank property of $P$ on the first perturbation, we have that with probability $\ge 1-\gamma$, $J(x')$ has at least $k$ singular values of magnitude $\ge \tau$ (the first scale is in $[\rho/2,\rho]$). As long as we choose $z$ such that
\[ z \rho < \frac{\tau}{4Mnk},  \]
we can apply the Lemma at each fixed $x'\in\mathcal B$ satisfying the rank bound to conclude that for any $\eps >0$,

\[
\forall y,\quad \Pr[\norm{P(\xtil)-y}<\eps,\ \xtil\in\mathcal B\mid x']
\le \left(\frac{2\eps}{\tau z\rho}\right)^k
  +\left(\frac{32Mz\rho nk}{\tau}\right)^{k/2}.
\]
It suffices to consider $0<h<1$. Let $0 < f < 1$ be a parameter that we will fix shortly. We first choose $z = \min\left\{\frac12,\frac{\tau f}{256Mnk\rho}\right\}$, so that the latter term above is at most $(f/8)^{k/2}$. Then, we pick $\eps = \frac{f\tau z\rho}{8}$, so that the former term becomes $(f/4)^k$. Putting these together, we have that for all $f \in (0,1)$,

\[
\forall y,\quad \Pr[\norm{P(\xtil)-y}<\eps,\ \xtil\in\mathcal B\mid x']
\le f^{k/2}.
\]
Writing $h = f^2$ and incorporating the failure probability of the Jacobian rank guarantee and the escape probabilities of $x'$ and $\xtil$, the theorem follows. Each escape probability is at most $\xi$: the first perturbation has a smaller scale than the total perturbation, and $\mathcal B$ is convex and contains $x$.
\end{proof}

\subsection{Jacobian rank property for Khatri Rao products}\label{sec:khatri-rao-jacobian}
As the first application, let us use the machinery from the previous sections to prove the following.

\begin{theorem}\label{thm:application:kr2}
    Suppose $U, V \in \R^{n \times m}$ and suppose their entries are independently perturbed (by Gaussians $\calN(0,\rho^2)$) to obtain $\tilde{U}$ and $\tilde{V}$. Suppose for some $L\ge1$ that $\Pr[\|\tilde U\odot\tilde V\|>L]\le\eta$. Then whenever $m \le n^2/C$ for some absolute constant $C$, we have

    \[ \sigma_{\min}(\tilde{U}\odot\tilde{V})\ge\frac1{\poly(n,L,1/\rho)}, \]
    with probability $1-\exp(-\Omega(n)){}-\eta$.
\end{theorem}

Note that when the unperturbed vectors have polynomially bounded norms, the result is stronger in terms of the success probability than the main result of~\cite{BCMV} and matches the result of~\cite{ADMPSV18}. The following lemma is the main ingredient of the proof, as it proves the CAA property for $\tilde{U} \odot \tilde{V}$. Theorem~\ref{thm:application:kr2} then follows immediately from Theorem~\ref{thm:caa-to-sigma}. The CAA constant $c=n/(24\sqrt m)$ gives failure probability $e^{-c\sqrt m/8}+\eta=e^{-n/192}+\eta$.

\begin{lemma}\label{lem:caa-for-kr2}
Suppose $\alpha \in \R^m$ is a unit vector at least $k$ of whose coordinates have magnitude $\ge \delta$. Let $U, V$ be arbitrary (as above), and let $\tilde{U}$ and $\tilde{V}$ be independent $\rho$ perturbations. Define $P(\tilde{U}, \tilde{V}) = \sum_i \alpha_i \tilde{u}_i \otimes \tilde{v}_i$.  Then there is an absolute constant $C_0>0$ such that for all $h>0$, we have
\[ \Pr\left[ \norm{P(\tilde{U},\tilde{V})} < \delta^2 h \cdot \frac{\rho^2}{C_0n^2k^2} \right] < \exp \left( -\frac{1}{24} \min \Big\{kn, n^2\Big\} \log(1/h) \right). \]
\end{lemma}

\paragraph{Remark.} To see why this satisfies the CAA property (hypothesis of Theorem~\ref{thm:caa-to-sigma}), note that as long as $m < n^2/C$ for a sufficiently large (absolute) constant $C$, the term $\frac{n\min\{k,n\}}{24} \ge \frac{n}{24\sqrt m} \min(m, k m^{1/2})$, thus it satisfies the condition with $\beta = 1/2$, with $p=2$, $c=n/(24\sqrt m)\ge16$, and $\lambda=C_0n^2m^2/\rho^2$.

\begin{proof}
Recall that $P$, as a function of $\tilde U$, is a map that has $mn$ variables whose output is an $n^2$ dimensional vector. We will argue (using the $k$ large coordinates of $\alpha$) that its Jacobian has sufficiently many nontrivial singular values with high probability.
To see this, observe that for a single term $\alpha_i (\tilde{u}_i \otimes \tilde{v}_i)$, the Jacobian (with respect to only $u_i$ variables) is simply an $n^2 \times n$ matrix, structured as follows: in the $j$th column, the $j$th ``block'' of size $n$ is $\alpha_i \tilde{v}_i$, and the rest of the entries are $0$. This holds for all $j$. Thus if $|\alpha_i| \ge \delta$, this matrix has $n$ singular values $\ge \delta \norm{\tilde{v}_i}$.

Next, if we take $\sum_i \alpha_i (\tilde{u}_i \otimes \tilde{v}_i)$, as the set of variables is different for every $i$, the overall Jacobian is the concatenation of the matrices described above (which is an $(n^2 \times nm)$ matrix). Thus, suppose we consider a set $I$ of $q=\min\{k,n\}$ indices with $|\alpha_i|\ge\delta$ and form the matrix (call it $W$) with columns $\{ \alpha_i \tilde{v}_i \}_{i \in I}$. If we argue that $W$ has $k'$ large singular values, then the structure above will imply that the Jacobian has $nk'$ large singular values.

Condition on all columns outside $I$. Below, $P$ denotes the selected sum $\sum_{i\in I}\alpha_i\tilde u_i\otimes\tilde v_i$, and we prove the estimate uniformly over a target $y\in\R^{n^2}$ to account for the fixed unselected sum.

Thus, let us focus on $W$. Put $s=\lceil q/2\rceil$. Lemma~\ref{lem:perturbed-sing-value-basic} now  shows that for any $h\in(0,1)$, $W$ has at least $s$ singular values of magnitude $\ge h\delta\rho/(1024\sqrt e)$, with probability at least  $1-(h/256)^{nq/8}$. Thus, the Jacobian has $ns$ singular values of magnitude $\ge h\delta\rho/(1024\sqrt e)$ (with the same probability). Moreover the Jacobian is $1$-Lipschitz in operator norm since $\alpha$ is a unit vector.
Adding the derivatives with respect to the selected $v_i$ cannot decrease these singular values, so the same lower bound holds for the joint Jacobian.
Thus, we can apply Theorem~\ref{thm:jacobian-to-antic}, with parameters $\tau=h\delta\rho/(2048\sqrt e)$, $M=1$, input dimension $2nq$, $\gamma=(h/256)^{nq/8}$, $\mathcal B=\R^{2nq}$, small-ball parameter set to $h/256$, and rank $ns$. We obtain, for sufficiently large $C_0$,

\[
\Pr\left[\|P-y\|<\frac{\delta^2 h^3\rho^2}{C_0n^2q^2}\right]
\le (h/256)^{nq/8}+(h/256)^{ns/4}< h^{nq/8}.
\]
Replacing $h^3$ with $h$ and using $q\le k$ and uniformity in the fixed noise, the lemma follows.
\end{proof}

\paragraph{Higher order Khatri-Rao products.}
The Jacobian property used to show Lemma~\ref{lem:caa-for-kr2} can be extended to higher order Khatri-Rao products. We outline the argument for any fixed order $d\ge3$: suppose $\{\tilde u_i^{(1)},\ldots,\tilde u_i^{(d)}\}$ are independently $\rho$-perturbed vectors, and define $P = \sum_{i} \alpha_i \bigotimes_{\ell=1}^d\tilde u_i^{(\ell)}$, for a unit coefficient vector $\alpha$ that has $k$ coordinates of magnitude $\ge \delta$. Select $q=\min\{k,\lfloor cn^{d-1}\rfloor\}$ of these coordinates, indexed by $I$, for sufficiently small $c>0$ depending on $d$. Now, the Jacobian (with respect to the $\tilde u^{(1)}$ variables) will have as a sub-matrix, the matrix $W$ whose columns are $\{\alpha_i\bigotimes_{\ell=2}^d\tilde u_i^{(\ell)}\}$ for $i\in I$. Now, we need to show that $W$ has at least $q/2$ large singular values, for $q\le cn^{d-1}$. Instead of a direct argument, we can now use our result for Khatri-Rao products of $d-1$ matrices!

We argue by induction on $d$, starting from Lemma~\ref{lem:caa-for-kr2}. Let $R\ge1$ bound the norms of the unperturbed vectors. Let $W_0$ have columns $\bigotimes_{\ell=2}^d\tilde u_i^{(\ell)}$ for $i\in I$. Put $s=\lceil q/2\rceil$. Since $W=W_0\diag(\alpha_I)$ and $|\alpha_i|\ge\delta$, we have $\sigma_s(W)\ge\delta\sigma_s(W_0)$. Thus we only need a median-singular-value bound for the unweighted matrix $W_0$.

If $\sigma_s(W_0)<a$, the column-subset argument in Lemma~\ref{lem:perturbed-sing-value-basic} gives $s-1$ columns $J$ with $\|\Pi W_0\|_F<qa$, where $\Pi$ projects off their span. Condition on these columns and group a constant fraction of the others into independent blocks of size at most a small constant times $n^{d-2}$ and comparable to $\min\{q,n^{d-2}\}$. For each block, sum its projected columns and run the lower-order Jacobian argument on that sum, starting from the proof of Lemma~\ref{lem:caa-for-kr2}. The projection loses at most $s-1$ large Jacobian singular values, and $q\le cn^{d-1}$ leaves a constant fraction of the $\Omega_d(n|B|)$ directions for a block $B$. The block coefficients are all $1$; after normalization they are $1/\sqrt{|B|}$, so the $p=2$ loss is only polynomial in $|B|$. The residual bound forces the norm of each projected block sum to be below $q n^{(d-2)/2}a$; choose $a=t/\poly_d(n,R,1/\rho)$. Multiplying the block probabilities and summing over at most $2^q$ choices of $J$ gives failure probability $t^{\Omega_d(nq)}$ for small enough $t$.

Condition on the terms of $P$ with $i\notin I$ and apply Theorem~\ref{thm:jacobian-to-antic} to the sum over $i\in I$ with rank $ns$ and threshold $\tau=\delta a$ at scales in $[\rho/2,\rho]$. On a ball about the unperturbed input, the Jacobian has Lipschitz bound $M=\poly_d(n,R,\log(1/t))$, and the perturbed input lies outside this ball with probability at most $t^{\Omega_d(nq)}$. Taking the theorem parameter to be $t$ and reparametrizing gives CAA with $\beta=1/d$, $p=2$, CAA constant $\Omega_d(n/m^{1/d})$, and $\lambda=\poly_d(n,R,1/\rho)$ for $m\le c_dn^d$, with $c_d>0$ sufficiently small.

 Then, assuming the order-$d$ Khatri-Rao matrix has norm at most $L\ge1$ with probability $1-\eta$, we can complete the proof as before, except now we apply Theorem~\ref{thm:caa-to-sigma} with $\beta = 1/d$ and $p=2$. We omit the details as they are identical to Theorem~\ref{thm:application:kr2}. The resulting least singular value bound is inverse-polynomial in $n,R,L,1/\rho$, with failure $\exp(-\Omega_d(n))+\eta$.

\paragraph{Other applications.}  In Section~\ref{sec:kronecker} we saw other natural candidates for $\calM$ including, most significantly, matrices obtained by applying a linear operator to a Kronecker product of matrices on some base variables. It is natural to ask if we can prove these results using the Jacobian based techniques we saw above. It turns out that this is possible for second order Kronecker products (here the CAA property corresponds to amplification for test ``matrices'' $\alpha$ that have large rank). But the method is not strong enough to handle higher order Kronecker products. We omit the details, since we can handle the general case using techniques in Section~\ref{sec:kronecker}.

\section{Application: Certifying Quantum Entanglement and Linear Sections of Varieties}
\label{sec:entangled}

We consider the setting in the work of Johnston, Lovitz and Vijayaraghavan~\cite{JLV2023}, where we are given a conic algebraic variety $\X$ and a linear subspace $\U$. The subspace $\U$ is specified by a basis
while the variety $\X$ is specified by a set of polynomials that cut it out. Conic varieties (or equivalently, projective varieties) are closed under scalar multiplication and are cut out by homogeneous polynomials, which can further be chosen to all have the same degree $d$. In \cite{JLV2023}, they give a polynomial time algorithm that  certifies that the intersection $\U \cap \X$ is trivial i.e., $\U \cap \X=
\{0\}$ for a {\em generic} subspace $\U$ up to a certain dimension $m$. In this section, we prove a robust analogue of this statement in the smoothed analysis setting.  

In the smoothed setting, the subspace $\tcU$ is spanned by $\rho$-perturbed vectors $\tilde{u}_1, \dots, \tilde{u}_m$, with $\forall i \in [m], ~\tilde{u}_i=u_i + z_i$ where $\norm{u_i}\le 1$ and $z_i \sim_{i.i.d} N(0,\rho^2)^n$ with $0<\rho\le1$. The subspace $\tcU$ is specified in terms of any orthonormal basis $\hat{u}_1, \dots, \hat{u}_m$. The variety is specified by a set of degree-$d$ homogenous polynomials that cut out the variety $\X$.  Our goal is to certify that every element of $v \in \X$ (with $\norm{v}=1$) is far from the subspace $\tcU$.

\begin{theorem}
\label{thm:certifyingsections}
Let $\X \subseteq \field^n$ be an irreducible variety cut out by $p= \delta \binom{n+d-1}{d}$ linearly independent homogeneous degree-$d$ polynomials $f_1,\dots, f_p \in \field[x_1,\dots, x_n]_d$, for constants $d \geq 2$ and $\delta \in (0,1)$. 
There exists a constant $c_d>0$
(that only depends on $d$) such that for a randomly $\rho$-perturbed subspace $\tcU \subseteq \R^{n}$ of dimension $m \le c_d \cdot \delta n$ as described above, 
we have that with probability $1-\exp(-\Omega(n))$, the algorithm in Figure~\ref{fig:alg} on input $f_1, \dots, f_p$ and a basis $u_1, \dots, u_m$ for $\tcU$ certifies in polynomial time (i.e.,  $(n/\rho)^{O(d)}$)  that 
\begin{equation}
    \forall v \in \X \text{ with } \norm{v}=1, ~ \textrm{dist}(v, \tcU) \coloneqq \inf_{u \in \tcU} \norm{u-v}_2 \ge \frac{\rho^d}{n^{O(d)}}. 
\end{equation}
\end{theorem}

This theorem gives a robust analog of the genericity statement in \cite{JLV2023} that certifies trivial intersection with a generic subspace (Theorem 2 with $s=0$). We remark that \cite{JLV2023} also considers a version of the problem where there are also some generic elements of $\X$ planted in the subspace $\tcU$. While it is natural to think of generic elements of $\X$ (using the induced Zariski topology on $\X$), it is not clear how to define an appropriate smoothed analysis model that captures the planted setting. This is an interesting research direction that is beyond the scope of this work.

We use the algorithm of Johnston, Lovitz and Vijayaraghavan~\cite{JLV2023}, which is based on Hilbert's projective Nullstellensatz certificates. Recall that $\X$ is a conic variety that is cut out by a finite set of homogeneous degree-$d$ polynomials $f_1,\dots, f_p \in \field[x_1,\dots,x_n]_d$. Viewing $f_1, \dots, f_p \in (\R^n)^{\otimes d}$ as vectors in the dual space, we consider the map
\begin{align}
\label{eq:phi}
\Phi_{\X}^d: (\R^n)^{\otimes d} &\rightarrow \field^{p}\\
v &\mapsto (f_{1}(\Sym_d v),\dots,f_{p}(\Sym_d v))^{\top}.
\end{align}
Note that for rank-1 \(v\), i.e. \(v = x^{\otimes d}\) for some \(x \in \mathbb{R}^n\), \(\Phi^d_{\mathcal{X}} = \mathbf{0}\) exactly when \(x\) lies in the variety \(\mathcal{X}\).   
By picking an orthonormal basis for the vector space spanned by the dual vectors $f_1, \dots, f_p$, we can assume without loss of generality that the operator $\Phi_\X^d$ is an orthogonal projection matrix with  rank $p$. Moreover we can assume that the given basis $u_1, \dots, u_m$ for $\tcU$ is an orthonormal basis.

\begin{figure}[htbp]
\begin{center}
\fbox{\parbox{0.98\textwidth}{

\begin{center}\textbf{\Large{Algorithm certifying that $\tcU$ is far from $\X$.}}\end{center}
\vspace{10pt}

\textbf{Input:} A basis $\{\hat{u}_1,\dots, \hat{u}_{m}\}$ for a linear subspace $\tcU \subseteq \R^n$, and a collection of homogeneous degree-$d$ polynomials $f_1,\dots, f_p$ that cut out a conic variety $\X\subseteq \R^n$.
\begin{enumerate}
\item Compute the least singular value of the following matrix
\begin{equation}\label{eq:lininalg}
s=\sigma_{\binom{m+d-1}{d}}\Big(\Phi_{\X}^d(\hat{u}_{i_1} \otimes \dots \otimes \hat{u}_{i_d}) : 1 \le i_1 \le \dots \le i_d \le m \Big).
\end{equation}
Set $\eta=s/(\sqrt{2}\,d)$.
\item If $\eta>0$,  output: ``$\tcU$ is at least $\eta$-far from $\X$.''
\item Otherwise, output: ``Don't know.''
\end{enumerate}
}
}
\end{center}
\phantomcaption
\label{fig:alg}
\end{figure}

\begin{proof}

    From standard random matrix theory with probability $1-\exp(-\Omega(n))$ we have for some constant $c_1>0$, $\sigma_{m}(\tilde{U}) \ge c\rho/(n^{c_1})$ and $\sigma_1(\tilde{U}) \le O(\sqrt{n}(1+\rho))$. Conditioned on the above event, every unit vector $u \in \tcU$ can be expressed as
    \begin{align}\label{eq:cert:temp1}
        u= \sum_{i=1}^m \alpha_i \hat{u}_i, \text{ for some } \alpha \in \R^m \text{ where } \norm{\alpha}_2=1.
    \end{align}

    Moreover, from Theorem~\ref{thm:kronecker:sym}, with probability at least $1-\exp(-\Omega(n))$ we have for some $c_d>0$ (that is constant for constant $d$) that
    \begin{equation}\label{eq:cert:temp2}
        \sigma_{\binom{m+d-1}{d}}\Big(\Phi_{\X}^d ~ \widetilde{U}^{\veep d}\Big) \ge \frac{c_d \rho^d}{n^{O(d)}}. 
    \end{equation}
    When we express the degree-\(d\) symmetric products of the \(\tilde u_i\) in the \(\hat u_i\) basis, coefficient norms can grow by at most a factor of \(\sqrt{d!}\,\|\tilde U\|^d\).
    Thus \eqref{eq:cert:temp2} gives
    \[
    s\ge\frac{c_d\rho^d}{\sqrt{d!}\,\|\tilde U\|^d n^{O(d)}}
    \ge\frac{\rho^d}{n^{O(d)}}.
    \]
    We condition on all the above high probability events (about the least singular values) that hold with probability $1-\exp(-\Omega(n))$. 
    
    Consider any vector $v \in \X \subset \R^n$ such that $\norm{v}=1$. Let $u^* \in \tcU$ be a closest unit vector to $v$ in the subspace $\tcU$. On the one hand, from \eqref{eq:cert:temp1} applied with $u^*$,
    \begin{align*}
        \Phi_{\X}^d((u^*)^{\otimes d})&= \sum_{i_1, i_2\dots,i_d=1}^m (\alpha_{i_1}\dots \alpha_{i_d})\cdot \Phi_{\X}^d \big( \hat{u}_{i_1} \otimes \dots \otimes \hat{u}_{i_d}\big)  \\
        &=  \sum_{1 \le i_1 \le i_2 \le  \dots,i_d\le m} \beta_{i_1, i_2, \dots, i_d}\cdot \Phi_{\X}^d \big( \hat{u}_{i_1} \otimes \dots \otimes \hat{u}_{i_d}\big),
    \end{align*}
where $\beta_{i_1, \dots, i_d}=c_{i_1,\dots,i_d} \alpha_{i_1}\dots \alpha_{i_d}$ and $c_{i_1,\dots,i_d}$ is a constant coefficient in $[1,d!]$.\footnote{More precisely, \(c_{i_1, \dots, i_d}\) is the number of entries in the tensor/dual form that correspond to the relevant monomial in the polynomial.  Thus, $c_{i_1,\dots,i_d}=d!/(t_1! t_2! \dots t_\ell!)$ where $\ell$ is the number of distinct indices among $i_1, \dots, i_d$, and $t_1, \dots, t_\ell$ are the frequencies of these $\ell$ distinct indices.} 
Also $1\le \norm{\beta} \le \sqrt{d!} \cdot \norm{\alpha}_2$. Hence, from the definition of $s$,
    \begin{align}
 \Big\|\Phi_{\X}^d((u^*)^{\otimes d})\Big\|_2&= \Big\| \sum_{1 \le i_1 \le i_2 \le  \dots,i_d\le m} \beta_{i_1, i_2, \dots, i_d} \Phi_{\X}^d \big( \hat{u}_{i_1} \otimes \dots \otimes \hat{u}_{i_d}\big)\Big\|_2 \nonumber\\
 &\ge s\cdot \norm{\beta}_2 \ge s. \nonumber
 \end{align}
Since $u^*$ is a closest unit vector, $\|v-u^*\|^2=2(1-\langle v,u^*\rangle)\le2(1-\langle v,u^*\rangle^2)=2\mathrm{dist}(v,\tcU)^2$.
On the other hand, we have $\Phi_{\X}^d (v^{\otimes d})=0$ since $v \in \X$. As $\norm{\Phi_{\X}^d}\le 1$ by assumption, we get the following sequence of inequalities
\begin{align*}
\mathrm{dist}(v,\tcU)\ge \frac{1}{\sqrt{2}\,d} \Bignorm{v^{\otimes d}-(u^*)^{\otimes d}} \ge \frac{1}{\sqrt{2}\,d} \Bignorm{\Phi_{\X}^d v^{\otimes d}- \Phi_{\X}^d (u^*)^{\otimes d}} \ge \frac{s}{\sqrt{2}\,d}=\eta.
\end{align*}
This proves the theorem. 
\end{proof}

\subsection{Certifying quantum entanglement}

We can instantiate the above theorem with specific choices of the variety $\X$ to get algorithms that certify that a smoothed subspace is robustly entangled i.e., it is far from any non-entangled state, for different notions of entanglement. In what follows, we restrict to the real domain for all of our statements and proofs. However, quantum states are defined over the complex domain. While we expect the ideas to extend in a natural way to the complex domain as well, we do not handle this setting in this paper. 

We start with the bipartite setting of dimension $n_1 \times n_2$, which is captured by matrices $\R^{n_1 \times n_2}$. The set of {\em separable} states is captured by the variety of rank-$1$ matrices:
\begin{equation} \label{eq:X1}
\X_1 = \set{M \in \R^{n_1 \times n_2}: ~ \rank(M) \le 1}. 
\end{equation}
A pure state $v$ that is non-separable i.e., $v \notin \X_1$ is said to be {\em entangled}. For any $\varepsilon>0$ we say that a subspace $\U \subset \R^{n_1 \times n_2}$ is said to be $\varepsilon$-robustly entangled if every unit vector in $\X_1$ is at least $\epsilon$ far from the subspace $\U$ i.e., 
\begin{equation} \text{($\varepsilon$-robust entanglement)}~\qquad~\label{eq:robust:X1}
\text{dist}(\U, \X_1) = \inf_{\substack{v \in \X_1: \norm{v}_2=1,\\ u \in \U}} \norm{v -u}_2 \ge \varepsilon. 
\end{equation}

More generally, the determinantal variety $\X_r$ corresponds to states that have {\em Schmidt rank} at most $r$ where
\begin{align}\label{eq:Xr}
&\X_r=\{ M \in \R^{n_1 \times n_2}: ~ \rank(M) \le r\}.\\
\text{($\varepsilon$-robust $r$-entanglement)}~\qquad~\label{eq:robust:Xr}&
 \text{dist}(\U, \X_r) = \inf_{\substack{v \in \X_r: \norm{v}_2=1,\\ u \in \U}} \norm{v -u}_2 \ge \varepsilon.
\end{align}
captures the $\varepsilon$-robustly $r$-entanglement of a subspace, for an $\varepsilon>0$. Such subspaces are only close to highly entangled states. Entangled subspaces corresponds to the setting when $r=1$. The variety $\X_r$ is cut out by $p=\binom{n_1}{r+1} \binom{n_2}{r+1}$ linearly independent homogenous polynomials of degree $r+1$;\footnote{corresponding to all the determinants of the $(r+1) \times (r+1)$ submatrices being $0$} see \cite{JLV2023}. The following corollary then follows immediately from Theorem~\ref{thm:certifyingsections}.

\begin{corollary}
     \label{cor:XR}
 Let $n_1, n_2$ be positive integers, let $r < \min \{n_1,n_2\}$ be a positive integer. 
There exists constants $c_r, c'_r>0$
(that only depends on $r$), an absolute constant $c>0$ and an algorithm that given a randomly $\rho$-perturbed subspace $\tcU \subseteq \R^{n_1 n_2}$ of dimension $m \le c_r  \cdot n_1 n_2$, 
runs in $(n_1 n_2/\rho)^{O(r)}$ time and certifies with probability at least $1-\exp(-\Omega(n_1n_2))$ that $\tcU$ is $\eta$-robustly $r$-entangled for $\eta=c'_r \rho^{r+1}/(n_1 n_2)^{c r}$ i.e.,
\begin{equation}
    \text{dist}(\tcU, \X_r) \ge \frac{c'_r \rho^{r+1}}{(n_1 n_2) ^{cr}}. 
    \end{equation}

\end{corollary}

Note that when $r=1$, we get Corollary~\ref{cor:X1} about robustly certifying entanglement of smoothed subspaces.
This is the robust generalization of Corollary 28 in \cite{JLV2023} which certifies that $\tcU \cap \X_r = \{0\}$ for a generic subspace $\tcU$  (this can be seen as a special case when $\eta  \to 0$). Our result gives a way of certifying a lower bound on the distance of a subspace from the set of rank-$1$ matrices for smoothed subspaces. This certification problem is closely related to the best separable state problem~\cite{HM10} which also relates this question to several other important questions in quantum information theory and polynomial optimization.  

\paragraph{Comparison to Barak, Kothari and Steurer~\cite{barak2017quantum}}
Barak, Kothari and Steurer~\cite{barak2017quantum} give an algorithm for an $\epsilon$-promise version of the entanglement certification problem, where given an arbitrary subspace $\U \in \R^{n \times n}$, the goal is distinguish between
\begin{itemize}
\item {\bf YES: } There is a unit vector $v \in \X_r$ that also lies in $\U$.
\item {\bf NO: } For all unit vectors $v \in \X_r$, $\norm{v-u} \ge \epsilon$ for all $u \in \U$.  
\end{itemize}
The algorithm in \cite{barak2017quantum} gives a $2^{O(\sqrt{n})/\epsilon}$ time algorithm to solve the above promise problem. Harrow and Montanaro ~\cite{HM10} presented evidence through conditional hardness results that polynomial time algorithms may not exist even for constant $\epsilon>0$  in the worst-case, i.e., for arbitrary subspaces. In contrast, our algorithm gives polynomial time algorithms for smoothed subspaces up to dimension $c \cdot n^2$ (for some constant $c>0$) even for $\epsilon$ that is inverse polynomially small.

We can also use Theorem~\ref{thm:certifyingsections} with other choices of $\X$ to get robust analogs of the other certification results in \cite{JLV2023}. These include multi-partite entanglement notions like complete entanglement and genuine entanglement which have been studied in quantum information. They follow by applying Theorem~\ref{thm:certifyingsections} along with the corresponding claims in \cite{JLV2023}. We state one such result for complete entanglement. Denote the set of separable order $d$-tensors 
\begin{equation} 
\X_{sep} = \{v_1 \otimes v_2 \otimes \dots \otimes v_d: v_1 \in \R^{n_1}, \dots v_d \in \R^{n_d} \}.
\end{equation}
A $\varepsilon$-robust completely entangled subspace $\U$ is one which is $\epsilon$-far from every unit vector in $\X_{sep}$. Again by using the fact that $\X_{sep}$ is cut out by $p=\binom{n_1 n_2 \dots n_d+1}{2}-\binom{n_1+1}{2} \cdot \dots \cdot \binom{n_d+1}{2}$ linearly independent homogenous polynomials of degree $2$ (see Section 2.3 of \cite{JLV2023}) we get the following corollary.

\begin{corollary}
     \label{cor:Xsep}
 Let $n_1, n_2, \dots, n_d$ be positive integers. 
There exists constants $c_d, c'_d>0$
(that only depends on $d$), an absolute constant $c>0$ and an algorithm that given a randomly $\rho$-perturbed subspace $\tcU \subseteq \R^{n_1 n_2 \dots n_d}$ of dimension $m \le c_d  \cdot n_1 n_2 \dots n_d$, 
runs in $\poly(n_1 n_2 \dots n_d/\rho)$ time and certifies with probability at least $1-\exp(-\Omega(n_1 \dots n_d))$ that $\tcU$ is $\eta$-robustly completely entangled for $\eta=c'_d \rho^2/(n_1 n_2\cdot n_d)^{c}$ i.e., 
\begin{equation}
    \text{dist}(\tcU, \X_{sep}) \ge \frac{c'_d \rho^{2}}{(n_1 n_2 \cdot n_d) ^{c}}.
    \end{equation}

\end{corollary}

\section{{Application: Decomposing Sums of Powers of Polynomials}}\label{sec:power-sums}\label{sec:bafna}

\errornote{
    In an earlier version of this manuscript, we claimed a smoothed analysis guarantee for the algorithm of \cite{BHKX2022} for power-sum decomposition, for the setting of sums of cubes of quadratics.  That claim relied on three propositions, one of which had an erroneous proof.  We include this overview section and point out the error where it occurs.  We \erroremph{cannot} claim a smoothed analysis guarantee for this algorithm in this work. 
}

\subsection{Overview of [Bafna, Hsieh, Kothari, Xu]}

One application of our techniques is to extend the results of Bafna, Hsieh, Kothari, and Xu \cite{BHKX2022} (which in turn build on \cite{GKS2020}) to the smoothed analysis setting.  In \cite{BHKX2022} they consider the problem of \emph{decomposing power-sum polynomials}.  In the most fundamental setting  they consider \(n\)-variate polynomials of the form 
\[\widehat{p}(\vx) = \sum_{t \le m} a_t(\vx)^3 + e(\vx),\]
where \(\vx = [x_1, \dots, x_n]\), each \(a_t(\vx)\) is a homogeneous quadratic polynomial, and \(e(\vx)\) is a polynomial of small norm. \cite{BHKX2022} also consider a more general setting where \(\widehat{p}(\vx) = \sum_{t \le m} a_t(\vx)^{3D} + e(\vx)\), where each \(a_t(\vx)\) is a homogeneous polynomial of degree \(K\), that is then taken to the \(3D\)th power for some integer \(D \ge 1\).  In this work we consider smoothed analysis for the setting when \(K = 2, D = 1\), as this is the simplest setting that captures many of their algorithmic ideas.  We do not handle the setting of general \(K, D\) in this work.  

The goal is to recover the underlying \(a_i(\vx)\)s from \(\widehat{p}(\vx)\).  \cite{BHKX2022} give an algorithm that is able to recover these components when there are up to \(m \sim \widetilde{O}(n)\) components, while withstanding noise of inverse polynomial magnitude, when each of the components \(a_i(\vx)\) is drawn \emph{randomly} from a mean-0 distribution.  

A natural setting for this problem is one in which each \(a_i(\vx)\) is \emph{perturbed}, rather than fully random.  In this setting, \cite{BHKX2022} are only able to show that their algorithm can withstand noise of inverse exponential magnitude.  Our contribution is to provide a new analysis of the random matrices that arise in their algorithm, and thus give a smoothed analysis guarantee for the algorithm of \cite{BHKX2022} that is robust to noise of inverse polynomial magnitude.  

In this section, we provide an overview of their algorithm and highlight the three main matrices that arise.  First, we rephrase the question about polynomials as a question about tensors.  That is, we identify each homogeneous quadratic polynomial \(a_t(\vx)\) with a symmetric coefficient matrix \(A_t\) such that \(a_t(\vx) = \vx^\top A_t \vx\).  Then, one way to represent \(p(\vx)\) is as a symmetric order-6 coefficient tensor $P_{\mathrm{sym}}$ defined by 
 
\[\mathrm{vec}(P_{\mathrm{sym}}) = \mathrm{Sym}_6 \Big( \mathrm{vec}\Big(\sum_{t \le m} A_t^{\otimes 3}\Big)\Big) ,\]
where \(\mathrm{Sym}_6\) is a linear operator that performs a 6-way symmetrization operation.\footnote{In terms of polynomials, \(\mathrm{Sym}_6\) essentially combines terms that are the same due to commutativity, i.e. \(x_1 x_2 x_3 x_4 x_5 x_6 = x_2 x_3 x_4 x_5 x_6 x_1\).  Note that, even though the underlying \(A_t\)s are (2-way) symmetric matrices, and each \(A_t^{\otimes 3}\) is taking a 3-way symmetric product, \(A_t^{\otimes 3}\) is \emph{not} necessarily 6-way symmetric.  One way to think of this is that the 2-way symmetry of the \(A_t\)s allows the variables of \(x_1 x_2 \cdot x_3 x_4 \cdot x_5 x_6\) to commute with their pairwise partners, and the 3-way symmetry of the tensor product allows the pairs to commute with each other.  However, this does not allow \(x_2\) and \(x_3\) to individually commute with each other, for example. Thus, the \(\mathrm{Sym}_6\) operator can change the structure of \(P\) significantly. }
Since we are given the input in the form of a polynomial, we only have access to this symmetrized form of the tensor.  

The main observation that the \cite{BHKX2022} algorithm is built on is that, if we could recover the asymmetric version of the tensor 
\[P_{\mathrm{asym}} = \sum_{t \le m} A_t^{\otimes 3},\] 
then we could use standard tensor decomposition techniques to recover the underlying \(A_t\)s.  In general, recovering an asymmetric tensor from a symmetric tensor is impossible, as there is a whole subspace of asymmetric tensors  
that could be mapped to the same symmetric tensor.  However, if the \(A_t^{\otimes 3}\) all belonged to a known generic low-dimensional subspace, then the symmetrization operator would actually be invertible over this subspace.  

In \cite{BHKX2022} they focus on recovering a basis for the subspace spanned by the (vectorized) \(A_t\)s.  If we recover such a basis, represented by the columns of a matrix \(C = [C_1, \dots, C_m]\), then we have that (the vectorized form of) \(P_\mathrm{asym}\) lives in the column span of \(C^{\veep 3}\), where \(\veep\) denotes the symmetrized Kronecker product.
We can then invert the symmetrization operator and retrieve \(P_\mathrm{asym}\), as long as the following claim holds. 

\begin{proposition}[Symmetrization invertible over Kronecker basis]
Given \(C \in \mathbb{R}^{n^2 \times m}\) representing an orthonormal basis for a subspace of $\rho$-perturbed symmetric matrices with Frobenius norms at most $1$ before perturbation, with \(m \le c n^2\) for some absolute constant \(c \in (0,1)\),
\(\mathrm{Sym}_6 ~ C^{\veep 3}\)
is \emph{robustly} invertible w.h.p.  That is, with probability at least $1-\exp(-\Omega(n))$,
\[\sigma_{\mathrm{min}} \left( \mathrm{Sym}_6 ~ C^{\veep 3} \right) \ge \poly\left( \rho, \frac{1}{mn} \right).
\]
\label{claim:symmetrization-invertible-over-Kronecker-basis}
\end{proposition}

Proposition \ref{claim:symmetrization-invertible-over-Kronecker-basis} follows from Theorem \ref{thm:kronecker:sym}.  First, we can consider \(C = A Q\) where \(Q\) is a basis change matrix, and \(A\) is the matrix of vectorized \(A_t\)s. Since \(C\) is orthonormal, \(\sigma_{\min}(Q)=1/\|A\|\ge1/\poly(n)\) with probability \(1-\exp(-\Omega(n))\). Therefore this same change of basis from $A$ to $C$, on third symmetric powers, has least singular value at least \(\sigma_{\min}(Q)^3/\sqrt{3!}\). The rank of \(\mathrm{Sym}_6\) is $\binom{n+5}{6}$, only a constant factor less than the dimension of the tensored space $(\binom{n+1}{2})^3$.
Thus this statement follows from Theorem \ref{thm:kronecker:sym}, considering the linear operator \(\mathrm{Sym}_6\) applied to symmetric lifts of the \(\rho\)-perturbed \(A_t\)s followed by the change of basis above.

To recover the subspace of the \(A_t\)s, \cite{BHKX2022} use the partial derivatives of \(p(\vx)\).  Specifically, the form of the second partial derivatives are 
\begin{align*}
    \frac{\partial^2}{\partial x_i \partial x_j} p(\vx) &= \frac{\partial^2}{\partial x_i \partial x_j} \sum_{t \le m} a_t(\vx)^3 
    = \sum_{t \le m} a_t(\vx) q_{t,i,j}(\vx),
\end{align*}
for some homogeneous quadratic polynomials \(q_{t,i,j}(\vx)\).  Now, we have a space \(\mathcal{U}\) spanned by these polynomial combinations of the \(a_t(\vx)\), and we would like to recover the underlying quadratic \(a_t(\vx)\)s from this space.  

\(\mathcal{U}\) is related to the space \(\mathcal{V}\), which spans \emph{all} quartic multiples of the \(a_t(\vx)\).  That is, 
\[ \mathcal{V} = \{ a_t (\vx) q(\vx)~:~ t \le m \text{ and } q(\vx) \text{ is a quadratic polynomial} \}.\]
Given \(\mathcal{V}\), \cite{BHKX2022} show a way to recover the span of the underlying \(a_t(\vx)\). So first, we would like to find \(\mathcal{V}\) starting from \(\mathcal{U}\).  By definition, we have that \(\mathcal{U} \subseteq \mathcal{V}\).  
However, it is \emph{not} true that \(\mathcal{U} = \mathcal{V}\).  This can be observed by dimension-counting: there are only \(\binom{n + 1}{2}\) partial derivatives, which cannot span the space of multiples which has dimension \(m \binom{n + 1}{2}\) for generic \(a_t(\vx)\).  

The \cite{BHKX2022} algorithm gets around this by projecting \(\mathcal{U}\) to a smaller dimensional space.  In particular, they show that by projecting \(\mathcal{U}\) onto \(\ell \in O(\sqrt{n})\) variables, they recover the span of all quartic multiples of the \(a_t(\vx)\) restricted to these \(\ell\) variables, i.e. the projected \(\mathcal{V}\).  We already know that the projected \(\mathcal{U}\) must be contained in the projected \(\mathcal{V}\).  Thus to show equality, we only need to show that the rank of the projected \(\mathcal{U}\) matches the rank of the projected \(\mathcal{V}\).  \cite{BHKX2022} do this by showing that a certain system of equations that recovers an element of the projected \(\mathcal{V}\) from a corresponding element of the projected \(\mathcal{U}\) is solvable.  This ends up boiling down to the following proposition.

\begin{restatable}[Projected \(\mathcal{U}\) same dimension as projected \(\mathcal{V}\)]{proposition}{projectedUequalsprojectedV}
Fix parameters $m,n,\ell \in \mathbb{N}$ with $\ell<n$ and $0<\rho\le1$, and let $\tilde{U}_1,\dots,\tilde{U}_m$ be a collection of $\rho$-smoothed symmetric $n \times n$ matrices which satisfy $\max_t {\|U_t\|_F} \leq L$. Let $M \in \R^{n \times \ell}$ be a column selection matrix and set $\tilde{S}_t = \tilde{U}_t M$ for each $t=1,\dots,m$. Also let $V \in \R^{n^2 \times m\binom{\ell+1}{2}+m}$ be the block matrix
\[
V := \begin{bmatrix}
\tilde{S}_1 \veep \tilde{S}_1 & \dots  & \tilde{S}_m \veep \tilde{S}_m & \mathrm{vec}(\tilde{U}_1) & \dots & \mathrm{vec}(\tilde{U}_m) 
\end{bmatrix}.
\]
Let $n_0=n-\ell$. Assume that, for some $\delta\in(0,1)$,
\[
m\left(\binom{\ell+1}{2}+1\right)-1
\le (1-\delta)\binom{n_0+1}{2},
\quad \text{and} \quad \ell\le c_2\delta n_0,
\]
where $c_2\le1$ is the degree-two constant in Corollary~\ref{corr:kron:blocks}.
Then there exists constants $c,c'>0$ (potentially depending on $\delta$ and $L$) such that with probability at least $1-\exp(-\Omega_\delta(n))$ we have
\[
\sigma_{m\binom{\ell+1}{2}+m} (V) \ge \frac{c' \rho^4}{n^c}. \]
\label{lem:projected-U-equals-projected-V}
\end{restatable}

Now, we shift our view to the case where we have \(\mathcal{V}\), the space of all quartic multiples of \(a_t(\vx)\).\footnote{In reality, we have access to the space of quartic multiples of the projected \(a_t(\vx)\).  However, we can repeat this step for various choices of the projection, and recover the \(a_t(\vx)\) restricted to different coordinates.  For simplicity of notation, we will still refer to the dimension of the (projected) polynomial as \(n\).}  \cite{BHKX2022} consider an equation of the form 
\begin{equation}
\sum_{t \le m} a_t(\vx) q_t(\vx) = 0,
\label{eq:generically-wellbehaved-polynomial-equation}
\end{equation}
where the \(a_t(\vx)\) are generic, and the \(q_t(\vx)\) are variables.  They note that since generic polynomials are \emph{irreducible}, the only solutions to this system have a certain structure.  In particular, the space of  solutions are spanned by solutions of the form 
\begin{equation}
\text{for } i \neq j: \quad  q_i(\vx) = a_j(\vx), \quad q_j(\vx) = - a_i(\vx), \quad q_k(\vx) = 0 \quad \forall k \neq i, j.
\label{eq:solutions-to-wellbehaved-polynomial-equation}
\end{equation}
and the dimension of the solution space is \(\binom{m}{2}\).  Note that solutions of this form always exist, even when the \(a_t(\vx)\) are not irreducible.  A key observation of \cite{BHKX2022} is that when \(a_t(\vx)\) are irreducible,
these are the \emph{only} solutions.  Now consider 
\begin{equation}
    \sum_{t \le m} a_t(\vx) q_t(\vx) = a_0(\vx) q_0(\vx),
    \label{eq:planted-polynomial-equation}
\end{equation} 
where \(a_0(\vx)\) is a fresh generic polynomial chosen by our algorithm.  The space of polynomials that have the form of the LHS of (\ref{eq:planted-polynomial-equation}) is precisely \(\mathcal{V}\).  Once we choose an \(a_0(\vx)\), the algorithm can also construct the space of all quartic multiples of \(a_0(\vx)\), which we denote \(\mathcal{V}_0\).  Now, consider all solutions to (\ref{eq:planted-polynomial-equation}) where \(q_0(\vx)\) is nonzero.  By the characterization of solutions in (\ref{eq:solutions-to-wellbehaved-polynomial-equation}), we know that if it is nonzero, \(q_0(\vx)\) must live in the span of the \(a_t(\vx)\).  Thus, we have that 
\[(\mathcal{V} \cap \mathcal{V}_0) / a_0(\vx) = \mathrm{span}\{a_t(\vx): t \le m\}.\]
To go from a generic guarantee to a smoothed guarantee, we must make the characterization (\ref{eq:solutions-to-wellbehaved-polynomial-equation}) of the solution space of this equation robust.  Thus, rather than only requiring the \(a_t(\vx)\) to be irreducible, we require the stronger condition that they are perturbed (smoothed).  
To formalize this, once again we represent our polynomials as symmetric tensors.  Denote the coefficient matrix of \(a_t(\vx)\) be \(A_t\), and the coefficient matrix of \(q_t(\vx)\) be \(Q_t\).  Then, we can write (\ref{eq:generically-wellbehaved-polynomial-equation}) as 
\[\mathrm{Sym}_4 ~\mathrm{vec}\left(\sum_{t \le m} A_t \otimes Q_t \right) = 0.\]
If we combine our \(Q_t\)s into one large vector, with the appropriate rearranging of indices\footnote{Think of \(Q = [Q_1 Q_2 \dots Q_m]\), where the \(Q_i\)s are vectorized into columns.  Then we vectorize \(Q\) in row-major order, rather than column major order.}, we can write this as 
\[\mathrm{Sym}_4 \left( I_{\binom{n+1}{2}} \otimes A \right) \mathbf{q} = 0,\]
where \(A\) is the matrix such that column \(i\) is (the vectorized) \(A_t\), and \(\mathbf{q}\) is the appropriate vectorization of the \(Q_t\)s.  Recall that we want to show that the dimension of the space of solutions \(\mathbf{q}\) is exactly \(\binom{m}{2}\).  Thus we can write the robust condition as following. In what follows $N_2 \coloneqq \binom{n+1}{2}$ is the dimension of the space of homogenous quadratic polynomials over $n$ variables.  

\errornote{
    The previously claimed proof of Proposition \ref{claim:not-too-many-solutions-to-polynomial-equation} was erroneous. We include the proposition statement for the sake of exposition of Proposition \ref{claim:symmetrization-invertible-over-Kronecker-basis} and Proposition \ref{lem:projected-U-equals-projected-V}, but we omit the erroneous proof from this version of the manuscript.    Because of this gap, the claimed result that the algorithm of \cite{BHKX2022} for decomposing sums of cubes of quadratics admits a smoothed analysis guarantee is \erroremph{not established}.  
    It is not clear whether the Proposition \ref{claim:not-too-many-solutions-to-polynomial-equation} is true as stated. \cite{BHKX2022} provide an analogous theorem when \(A\) is made up of \emph{random} (rather than smoothed) columns (up to a polylogarithmic factor in \(m\): i.e., they require \(m \le N_2 / \polylog(N_2)\)).
}

\begin{proposition}[Not too many solutions to polynomial equation system]
    The space of solutions to the above polynomial equation system has rank at most \(\binom{m}{2}\) in a \emph{robust} sense.  That is, let \(A \in \mathbb{R}^{N_2 \times m}\) be made up of perturbed symmetric (when viewed as matrices) columns. Then there exists constants $c,c',c''>0$ such that when $m \le c N_2$, with probability at least $1-\exp(-\Omega(n))$  
    \[\sigma_{mN_2  - \binom{m}{2}} \left( \mathrm{Sym}_4 ~ \left(I_{N_2 \times N_2} \otimes A \right) \right) \ge \frac{c'' \min\{\rho,1\}^2}{n^{c'}},\]
    where \(\otimes\) denotes the Kronecker product and $N_2 = \binom{n+1}{2}$.  
    \label{claim:not-too-many-solutions-to-polynomial-equation}
\end{proposition}

In \cite{BHKX2022} they prove analogues of propositions \ref{claim:symmetrization-invertible-over-Kronecker-basis}, \ref{lem:projected-U-equals-projected-V}, and \ref{claim:not-too-many-solutions-to-polynomial-equation},
for the case where the underlying matrices are fully random (mean-0), and they are able to show that these propositions hold with inverse-polynomial failure probability.  Our framework allows us to prove \erroremph{the first two} propositions for perturbed matrices, and show that the failure probability is \emph{exponentially} small.  We \erroremph{cannot} conclude a smoothed analysis guarantee for the algorithm of \cite{BHKX2022} in this work.

\subsection{Least singular value bounds for
Proposition~\ref{lem:projected-U-equals-projected-V} } 

\projectedUequalsprojectedV*

\begin{proof}

Without loss of generality assume that $\tilde{S}_t$ is the first $\ell$ columns of $\tilde{U}_t$. Also, we have by standard concentration bounds, $\max_t \norm{\tilde{U}_t}_F \le L + \rho \cdot \poly(n)$ with probability at least $1-\exp(-\Omega(n))$. For the rest of the proof, assume that the upper bound $L$ already includes the additive term $\rho \poly(n)$.

We first argue that each column of the form $\mathrm{vec}(\tilde{U}_t)$ has a large component that is orthogonal to the remaining columns of $V$. To this end, set
\[
V_t:=\begin{bmatrix}
\tilde{S}_1 \veep \tilde{S}_1 & \dots  & \tilde{S}_m \veep \tilde{S}_m  & \mathrm{vec}(\tilde{U}_1) & \dots & \mathrm{vec}(\tilde{U}_{t-1}) & \mathrm{vec}(\tilde{U}_{t+1}) & \dots & \mathrm{vec}(\tilde{U}_m) 
\end{bmatrix}
\]
We can restrict to considering only the lower-right $n_0\times n_0$ block of each vectorized matrix. By assumption, each $\tilde{S}_t$ is made up of the first $\ell$ columns of $\tilde{U}_t$. Write $B_t$ for the last $n_0$ rows of $\tilde{S}_t$. The smoothing in the lower-right block of $\tilde U_t$ is independent from the corresponding blocks in $V_t$, since the projected lift blocks depend only on $B_1,\ldots,B_m$.

Now, let $Q:\mathbb{R}^{n^2} \to \Sym^2(\R^{n_0})$ be the matrix that restricts to this block in orthonormal symmetric coordinates and let $\Pi_{QV_t}^\perp$ be the projection onto the orthogonal complement of the range of $QV_t$. Note that the matrix $QV_t$ has rank at most $m\binom{\ell+1}{2}+m-1$, so using Lemma \ref{lem:MultiVarGaussBall} conditional on $QV_t$ we have
\[
\mathbb{P}[\|\Pi_{QV_t}^\perp Q\mathrm{vec}(\tilde{U}_t)\| \geq \epsilon] \geq 1-\left(\frac{c''\epsilon}{\rho}\right)^r,
\]
where $r:=\binom{n_0+1}{2}-m\binom{\ell+1}{2}-m+1$ and $c''>0$ is an absolute constant. The probability that the above holds for all $t=1,\dots,m$ is then at least
\[
1-m\left(\frac{c''\epsilon}{\rho}\right)^r.
\]

Now suppose that $\sigma_{\min} (V) < \epsilon$. Then there must exist some test unit vector $\alpha \in \R^{m\binom{\ell+1}{2}+m}$ such that $\|V \alpha\| < \epsilon.$ Write $\alpha = \alpha^{(1)} \oplus \alpha^{(2)}$ where $\alpha^{1} \in \R^{m \binom{\ell+1}{2}}$ and $\alpha^{2} \in \R^{m}$. Intuitively, $\alpha^{(1)}$ contains the entries of $\alpha$ that are coefficients of columns of $V$ which are from the matrices $\tilde{S}_t \veep \tilde{S}_t$ while $\alpha^{(2)}$ contains the coefficients of the columns of $V$ of the form $\mathrm{vec}(\tilde{U}_t)$.

We consider two cases. In the first case suppose that $\|\alpha^{(2)}\| \leq \frac{\sqrt{\epsilon}}{2 \sqrt{m} L}$.  In this case we upper bound the probability that $\|V \alpha \| \leq \sqrt{\epsilon}$, which in turn upper bounds the probability that $\|V\alpha\| \leq \epsilon$. We have
\begin{align*}
\|V \alpha\| &= \left\|\begin{bmatrix}
\tilde{S}_1 \veep \tilde{S}_1 & \dots  & \tilde{S}_m \veep \tilde{S}_m
\end{bmatrix} \alpha^{(1)} + \begin{bmatrix}
\mathrm{vec}(\tilde{U}_{1}) & \dots & \mathrm{vec}(\tilde{U}_m) 
\end{bmatrix} \alpha^{(2)} \right\| \\
& \geq \left\|\begin{bmatrix}
\tilde{S}_1 \veep \tilde{S}_1 & \dots  & \tilde{S}_m \veep \tilde{S}_m
\end{bmatrix} \alpha^{(1)} \right\|-\sqrt{m}L \|\alpha^{(2)}\|.
\end{align*}
This would imply that
\[
\sigma_{\min} \left(\begin{bmatrix}
\tilde{S}_1 \veep \tilde{S}_1 & \dots  & \tilde{S}_m \veep \tilde{S}_m
\end{bmatrix}\right) < 2\epsilon+\sqrt{\epsilon} \leq 3\sqrt{\epsilon}.
\]
Taking $\epsilon = \frac{c' \min\{1,\rho^4\}}{n^c}$ as in the statement of the proposition, we can use Corollary \ref{corr:kron:blocks} on $[B_1^{\veep2}\ \cdots\ B_m^{\veep2}]$ with ambient dimension $n_0$, to conclude that the probability that such an $\alpha$ exists is at most $\exp\big(- \Omega_{\delta} (n)\big)$.

In the second case, we have that $\|\alpha^{(2)}\| > \frac{\sqrt{\epsilon}}{2 \sqrt{m} L}$. Therefore, $\alpha^{(2)}$ must have some component $\alpha^{(2)}_j$ with magnitude at least $|\alpha^{(2)}_j| \geq\frac{\sqrt{\epsilon}}{2 m L}$. It follows that
\[
 \|\Pi_{QV_j}^\perp Q \mathrm{vec}(\tilde{U}_j) \| \|\alpha^{(2)}_j\|= \|\Pi_{QV_j}^\perp Q V \alpha \|  \leq  \|V \alpha\| < \epsilon.
\]
Here the first equality follows from the fact that the only nonzero column of $\Pi_{QV_j}^\perp Q V$ is $\Pi_{QV_j}^\perp Q \mathrm{vec}(\tilde{U}_j)$ which implies $\Pi_{QV_j}^\perp Q V \alpha = \Pi_{QV_j}^\perp Q \mathrm{vec}(\tilde{U}_j) \alpha^{(2)}_j$. 
From this we obtain
\[
 \|\Pi_{QV_j}^\perp Q \mathrm{vec}(\tilde{U}_j) \|  < 2mL\sqrt{\epsilon}.
\]
Using the discussion above, the probability that there exists some $j \in [m]$ for which this inequality holds is at most $m\left(\frac{c''mL\sqrt{\epsilon}}{\rho}\right)^r$. Taking $\epsilon = \frac{c' \min\{1,\rho^4\}}{n^c}$ as in the statement of the proposition and recalling that $r \geq \delta\binom{n_0+1}{2}$, this can be written as $m\left(\frac{c''mL\sqrt{\epsilon}}{\rho}\right)^r = \exp\big(- \Omega_{\delta} (n)\big)$.

Combining these two cases and the failure probability of the norm bound, we see that the probability that there exists a unit vector $\alpha$ such that $\|V \alpha\| \leq \frac{c' \min\{1,\rho^4\}}{n^c}$ is bounded above by 
$\exp\big(- \Omega_{\delta} (n)\big)$,
which completes the proof.

\end{proof}

\section{Application: Subspace Clustering}\label{sec:subspace-clustering}
Very recently, Chandra, Garg, Kayal, Mittal, and Sinha~\cite{chandra2024Learning} extended the framework of~\cite{GKS2020} and obtained robust analogs of the corresponding learning algorithms. One of the applications they consider is subspace clustering, in which there are $s$ underlying hidden \(t\)-dimensional subspaces \(W_1, \dots, W_s\).  We see a set of points \(A\) such that \(A = A_1 \cup \dots \cup A_s \), where each \(A_i \subseteq W_i\).  Our goal is to recover the \(W_i\) from \(A\).  

\cite{chandra2024Learning} give an algorithm for this problem by reducing to \emph{robust vector space decomposition}.  
For a specific block \(i\), let \(V_i = \left[ \mathbf{v}_{i1}, \dots, \mathbf{v}_{it} \right]\) be a basis for \(W_i\).  
Treating each \(a^{\veep d}\) as a homogeneous degree-\(d\) polynomial, they take all first-order partial derivatives of all of the \(a^{\veep d}\).  Note that if \(a \in A_i\), then \(a^{\veep d} \in \mathrm{colspan}( V_i^{\veep d} )\), and the first order partial derivatives of \(a^{\veep d}\) are in \(\mathrm{colspan}(V_i^{\veep (d - 1)})\).  

This sets up the problem in the framework of vector space decomposition.  In vector space decomposition, we are given two vector spaces \(U\) and \(V\), where there are hidden decompositions \(U = U_1 \oplus \dots \oplus U_s\) and \(V = V_1 \oplus \dots \oplus V_s\), where \(\oplus\) denotes a direct sum.  We are also have a set of linear maps \(\mathfrak{B} = (B_1, \dots, B_\ell)\), where each \(B_j : U \rightarrow V\), while respecting that \(B_j U_i \subseteq V_i\) for all \(i\).  The goal of vector space decomposition is to recover the \(U_i\) (and therefore the \(V_i\)) given \(U, V,\) and \(\mathfrak{B}\).   

In the case of subspace clustering, we can take \(U\) to be the span of \(\{a^{\veep d} : a \in A\}\), \(\mathfrak{B}\) to be the set maps corresponding to first-order partial derivatives, and \(V\) to be the image space of \(\mathfrak{B}\) applied to \(U\).
Then, as long as subspaces \(\mathrm{colspan}\left(V_i^{\veep (d - 1)}\right) \) are linearly independent, and some non-degeneracy conditions on the set of points \(A\) are met, this set of linear maps allows us to recover the subspaces \(V_i\) via robust vector space decomposition.  

\cite{chandra2024Learning} assume two conjectures to prove that this algorithm succeeds in the smoothed framework.  The first is a smoothed version of the statement that the subspaces \(\mathrm{colspan}(V_i^{\veep (d - 1)}) \) are linearly independent.  We address this first conjecture in Section \ref{sec:garg-kayal-block-kronecker}.  The second conjecture is that for sufficiently many perturbed points \(A_i\) from subspace \(W_i\), the set \(\{a^{\veep d} : a \in A_i\}\) spans \(W_i^{\veep d}\) robustly.  We address this second conjecture in Section \ref{sec:garg-kayal-spanning}.

\subsection{Proof of Conjecture on Robust Linear Independence of Subspace Lifts}
\label{sec:garg-kayal-block-kronecker}

\begin{conjecture}[Conjecture 1.1 in \cite{chandra2024Learning}]
\label{conj:SubspaceBlockKron}
    Given $\delta>0$ and $d \in \mathbb{N}$, fix \(m < n\) (\(m\) is the dimension of the subspaces), \(s\) such that \(s \binom{m + d - 1}{d} \le \delta \binom{n + d - 1}{d}\). Additionally assume that $m \leq (1-\delta) c_d  n$, where $c_d$ is as in Corollary \ref{corr:kron:blocks}. 

    For each \(i \in [s]\), let \(V_i\)
    be an orthonormal basis for the subspace \(\widehat{W}_i\). 
    Each \(\widehat{W}_i\) is the subspace spanned by \(P_i + G_i \in \mathbb{R}^{n \times m}\), where \(P_i \in \mathbb{R}^{n \times m}\) is an arbitrary matrix with orthonormal columns, and \(G_i \sim \mathcal{N}(0, \rho^2 / n)^{n \times m}\) is drawn independently for each \(i \in [s]\).

    Let \(M = \left[ M_1 \ M_2 \ \dots \ M_s \right]\) be a matrix formed of \(s\) blocks, where 
    \[M_i =  V_i^{\veep d}.\]

    Then with probability at least $1-\exp (-\Omega_{d,\delta} (n))$, the least singular value of \(M\) is bounded as 
    \[\sigma_{\mathrm{min}}(M) = \sigma_{s \binom{m + d - 1}{d}}(M) \ge \mathrm{poly}(\rho, 1/n).\]
\end{conjecture}

To prove this conjecture we start with a lemma that lets us to reduce to the setting of Corollary \ref{corr:kron:blocks}.

\begin{lemma}
\label{lem:BlockMatrixChangeOfBasisBound}
 Fix $\ell,k,s \in \mathbb{N}$ and $\epsilon,\eta,L \geq 0$. Let $A_1,\dots,A_s,U_1,\dots,U_s \in \R^{\ell \times k}$ be matrices that for all $i = 1,\dots,s$ we have
 \[
    \mathrm{colspan} (A_i) = \mathrm{colspan} (U_i) \qquad \mathrm{and} \qquad \|A_i \|_2 \leq L \qquad \mathrm{and} \qquad \sigma_{\min} (U_i) \geq \eta.
 \]
Additionally assume that $\sigma_{\min} ([A_1 \ \dots \ A_s]) \geq \epsilon$. Then we have
\[
\sigma_{\min} ([U_1 \ \dots \ U_s]) \geq \frac{\eta \epsilon}{L}.
\]
\end{lemma}
\begin{proof}
Suppose towards a contradiction that 
 there is a unit vector $v = v_1 \oplus \dots \oplus v_s \in \R^{sk}$ such that
\[
\|[U_1 \ U_2 \ \dots \ U_s]v\| = \left\|\sum_{i=1}^s U_iv_i \right\| < \frac{\eta \epsilon}{L}.
\]
Then, since $\mathrm{colspan} (A_i) = \mathrm{colspan} (U_i)$ for each $i$, it follows that for each $i=1,\dots,s$, there exists a vector $z_i \in \R^{k}$ such that $U_i v_i = A_i z_i$. Set $z = z_1 \oplus \dots \oplus z_s$. Then we have
\[
\|[A_1 \ \dots \ A_s] z \| = \left\|\sum_{i=1}^s A_iz_i\right\| = \left\|\sum_{i=1}^s U_iv_i\right\| = \|[U_1 \ \dots \ U_s] v \| < \frac{\eta \epsilon}{L}.
\]
Furthermore, for each $i$ we find that
\[
\eta\|v_i\| \leq \|U_i v_i \| = \|A_i z_i \| \leq L \|z_i\|,
\]
hence
\[
\|z_i\| \geq \frac{\eta \|v_i\|}{L}, \qquad \mathrm{which \ implies} \qquad \frac{1}{\|z\|} \leq \frac{L}{\eta}. 
\]
We then obtain that
\[
\left\|[A_1 \ \dots \ A_s] \frac{z}{\|z\|} \right\| = \frac{1}{\|z\|} \|[A_1 \ \dots \ A_s] z \| < \frac{L}{\eta}\frac{\eta \epsilon}{L},
\]
which contradicts the assumption that $\sigma_{\min} ([A_1 \ \dots \ A_s]) \geq \epsilon$.
\end{proof}

We now prove Conjecture \ref{conj:SubspaceBlockKron}.

\begin{proof}
    For each $i=1,\dots,s$, set $\tilde{P}_i := P_i+G_i$. From Corollary \ref{corr:kron:blocks}, we have that with probability at least $1-\exp(-\Omega_{d,\delta}(n))$ that 
    \[
    \sigma_{\min} ([\tilde{P}_1^{\veep d} \ \dots \  \tilde{P}_s^{\veep d}]) \geq \frac{\rho^d}{\sqrt{s} n^{O(d)}}.
    \]
    Furthermore, for each $i$, we have with probability at least $1-\exp(-\Omega(n)),$ that
    \[
    \|\tilde{P}_i\| \leq 1+\rho c_1\sqrt{n},
    \]
    for some $c_1 > 0$. For the remainder of the proof, we condition on these events. 
    
    Now, since $\tilde{P}_i$ has the same range as $V_i$, we also have that $\tilde{P}_i^{\veep d}$ has the same range as $V_i^{\veep d}$ for each $i$. Furthermore, since each $V_i$ is an orthogonal matrix and since $V_i^{\veep d} = V_i^{\otimes d} \Sela$ and the matrix $\Sela$ has singular values in $[1/\sqrt{d!},1]$, we obtain that $\sigma_{\min} (V_i^{\veep d}) \geq 1/\sqrt{d!}$ for each $i=1,\dots,s$. Therefore, we can apply Lemma \ref{lem:BlockMatrixChangeOfBasisBound} to conclude that
    \[
    \sigma_{\min} ([V_1^{\veep d}\  \dots \ V_s^{\veep d}]) \geq \frac{\rho^d}{(1+\rho c_1 \sqrt{n})^d \sqrt{s(d!)} n^{O(d)}},
    \]
    as claimed.
\end{proof}

\subsection{Proof of the ``Spanning'' Conjecture of Subspace Lifts}
\label{sec:garg-kayal-spanning}
\cite{chandra2024Learning} also conjecture that if we have smoothed vectors $\{\tilde{u}_i\}_{i=1}^N$ in $\R^d$, and if $N $ is large enough, then the polynomials $ \langle \tilde{u}_i, x \rangle^r$ are all ``robustly'' linearly independent. Formally, they conjecture the following:

\begin{conjecture}[Conjecture 1.2 in \cite{chandra2024Learning}]\label{kayal-conjecture-2} Suppose $\sigma, \delta >0$ are constants, and let $r\ge 1$ be an integer. Let $\{\tilde{u}_i\}_{i=1}^N$ be independent $\sigma$-perturbations of some (arbitrary) vectors $\{u_i\}_{i=1}^N$ of norm $\le \Delta$. Suppose also that $N > (1+\delta) \binom{d+r-1}{r}$, and let $M$ be the $N \times \binom{d+r-1}{r}$ matrix whose $i$th row contains the vector of coefficients for the polynomial $\langle \tilde{u}_i, x\rangle^r$.  Then with high probability,
\[ \sigma_{\binom{d+r-1}{r}} (M) \ge \mathrm{poly}_r \left( \frac{\sigma \delta}{n \Delta} \right).  \]
\end{conjecture}

We give a short and self-contained proof of this conjecture, but with a weaker guarantee on $N$. Specifically, we show that the conclusion holds as long as
\begin{equation}\label{eq:conj-weaker} N \ge 2r \binom{ d+r-1}{r}.
\end{equation}

The main ingredient is the observation that for any unit vector $a \in \R^{\binom{d+r-1}{r}}$, $\langle M\su{i}, a \rangle$ is large enough with high probability, where $M\su{i}$ is the $i$th row of $M$, and the probability is over the choice of the randomness in $\tilde{u}_i$. I.e.,

\begin{lemma}\label{lem:carbery-wright-ext}
Suppose $a \in \R^{\binom{d+r-1}{r}}$ is a unit vector that we index using $r$-sized multisets $I$ of $[d]$. Let $x\in \R^n$ be an arbitrary vector, and let $\xtil$ be an entrywise $\sigma$ perturbation, for some $\sigma>0$. In other words, $\xtil_i := x_i + g_i$, where $g_i \sim \calN(0,\sigma^2)$. Finally, denote $f(\xtil) = \sum_I a_I \prod_{i \in I} \xtil_i$.  Then for any $y$ and any $\eps > 0$, we have
\[ \Pr[ | f(\xtil) - y | < \eps ] < \frac{c' r}{\sigma} \eps^{1/r},\]
for some absolute constant $c'$.
\end{lemma}
The Lemma is very similar to the classic Carbery-Wright inequality, but unfortunately it is not a direct consequence, because the Carbery-Wright inequality requires a bound on the \emph{variance} of $f$. Fortunately, it turns out that one can always obtain a lower bound on the variance in terms of the $\ell_2$ norm of the coefficient vector. 
\begin{proof}
As noted above, we first apply the main result from the recent work of~\cite{GLAZER2022109639} to lower bound the variance of $f$. To this end, note that $f(\xtil)$ is a homogeneous polynomial in the variables $\xtil_i$, whose distributions are independent, according to appropriately shifted Gaussians. For this setting,~\cite{GLAZER2022109639} prove that\footnote{We note that the part of Theorem 1 from~\cite{GLAZER2022109639} that we use assumes that the distribution of $\xtil_i$ is isotropic, but their proof applies directly even with shifts.}  
\[  \text{Var}(f(\xtil)) \ge \frac{\sigma^{2r}}{B^r} \cdot \sum_I a_I^2 ,   \]
where $B$ is a universal constant (which they upper bound by $2^{15}$, but is likely much smaller for Gaussians). Next, we apply Carbery-Wright inequality~\cite{CarberyWright}, which states that for any $y\in \R$ and $\eps>0$,
\[  \Pr[ | f(\xtil) - y| < \eps ] \le c r \frac{\eps^{1/r}}{\text{Var}(f)^{1/2r}}, \]
for some absolute constant $c$. Combining this with the lower bound on the variance and noting that $\norm{a}=1$, we obtain:
\[  \Pr[ | f(\xtil) - y| < \eps ] \le \frac{c' r}{\sigma} \eps^{1/r} \]
for some absolute constant $c'$. This completes the proof of the lemma.
\end{proof}

We can now prove the conjecture via a standard $\eps$-net argument. We also condition on the event that $\norm{M}$ is \emph{upper} bounded by $\text{poly}_r(n\Delta)$. This can easily be checked to hold with extremely high probability.

\begin{proof}[Proof of Conjecture~\ref{kayal-conjecture-2} assuming~\eqref{eq:conj-weaker}]
For any unit vector $a \in \R^{\binom{d+r-1}{r}}$, Lemma~\ref{lem:carbery-wright-ext} bounds the probability of one of the entries of $Ma$ being of small norm. Since the rows of $M$ all have independent randomness, we have, for any $\eps>0$,
\[ \Pr[ \norm{Ma} < \eps ]  < \left( \frac{c'r}{\sigma} \eps^{1/r} \right)^N. \]
Let us write $C = \frac{c'r}{\sigma}$, for convenience, so the RHS above becomes $(C\eps^{1/r})^N$; also, write $D = \binom{d+r-1}{r}$.

Let $\delta >0$ be a parameter we will choose later, and let $\calN_{\delta}$ be a $\delta$ net for unit vectors in $\R^{D}$. As is well known, we may assume that $|\calN_{\delta}| \le \left( \frac{4}{\delta} \right)^{D}$. By a union bound, we have that 
\[ \Pr[ \norm{Ma} \ge \eps ~ \forall a \in \calN_{\delta} ] \ge 1- |\calN_{\delta}| (C \eps^{1/r})^N.  \]
If we can ensure that this event occurs, we also have that for all unit vectors $a$, $\norm{Ma} \ge \eps - \delta \norm{M}$, where $\norm{M}$ is the spectral norm of $M$. First, we set $\delta = \frac{\eps}{2 \norm{M}}$, which ensures a lower bound of $\eps/2$ on the least singular value. Now for the high probability bound to hold, we require
\[ \left(  C \eps^{1/r} \right)^N \left( \frac{4}{\delta} \right)^D \ll 1 \iff \left(  C \eps^{1/r} \right)^N \left( \frac{8 \norm{M}}{\eps} \right)^D \ll 1.    \]
Now, by assumption~\eqref{eq:conj-weaker}, we have $N \ge 2rD$. Since adding more rows to $M$ only improves the least singular value, let us assume that $N = 2rD$. The equation above simplifies to $C^{2r} \eps \cdot 8 \norm{M} \ll 1$, which holds if we set $\eps = \frac{1}{16 C^{2r} \norm{M}} $. 

This completes the proof of the conjecture. Indeed, we see that the failure probability in this case is exponentially small in $D$.
\end{proof}

\section*{Acknowledgments}
\addcontentsline{toc}{section}{Acknowledgments}
Vaidehi Srinivas and Aravindan Vijayaraghavan were supported by the
National Science Foundation under Grant Nos. CCF-1652491, ECCS-2216970. Vaidehi Srinivas was also supported by the Northwestern Presidential Fellowship. Aditya Bhaskara was supported by the National Science Foundation under Grant Nos. CCF-2008688 and CCF-2047288. 
Any opinions, findings, and conclusions or recommendations expressed in this material are those of the authors and do not necessarily reflect the views of the National Science Foundation.
We thank the (anonymous) reviewers for their detailed feedback which helped improve the presentation, and we also thank Neeraj Kayal, Ankit Garg, and Benjamin Lovitz for helpful discussions.

\bibliographystyle{alpha}
\addcontentsline{toc}{section}{References}
\bibliography{references}

\appendix 
\section{Auxiliary Lemmas}

\begin{lemma}
\label{lem:MultiVarGaussBall}
Let $\tilde{u} \in \R^n$ be a $\rho$-smoothed vector and fix $\delta >0$. There exists a universal constant $c>0$ such that 
\[
Pr[\|\tilde{u} \|< \delta] \leq \Big(\frac{c \delta}{\rho} \Big)^n 
\]
\end{lemma}
\begin{proof}
The PDF of $\tilde{u}$ is bounded above by $\left(\rho\sqrt{2 \pi})^{-n}\right)$ from which we obtain that 
\begin{align*}
Pr[\|\tilde{u} \|< \delta] &\leq \int_{B(0,\delta)} \left(\rho\sqrt{2 \pi}\right)^{-n} dx = \left(\rho\sqrt{2 \pi}\right)^{-n} \mathrm{vol}\big(B(0,\delta)\big)  \\ &= \left(\rho\sqrt{2 \pi}\right)^{-n}  \frac{\left(\delta \sqrt{\pi}\right)^n }{\Gamma\left(\frac{n}{2}+1\right)} = \left( \frac{\delta}{\rho \sqrt{2}}\right)^n \frac{1}{\Gamma(\frac{n}{2}+1)}.
\end{align*}
Now, using Stirling's approximation we have 
\[
\Gamma\left(\frac{n}{2}+1\right) \geq \sqrt{\pi n}\left(\sqrt{\frac{n}{2e}}\right)^n.
\]
Substituting this probability in above gives the desired probability upper bound.
\end{proof}

\subsection*{Block leave-one-out bounds and least singular value of smoothed matrices}

\begin{lemma}
\label{lemma:BlockLeaveoneOut}
    Let $U_1,\dots,U_t \in \R^{n \times m}$ and for each $j=1,\dots,t$ let $\Pi^{\perp}_{-j}$ be the projection on to the orthogonal complement of 
    \[
    \mathrm{Ran} \left( \begin{bmatrix}
        U_1 & \dots & U_{j-1} & U_{j+1} & \dots U_t.
    \end{bmatrix}\right)
    \]
    Define $\ell_B(\{U_j\})= \min_{j} \sigma_{\min} (\Pi^\perp_{-j} U_j).$ Then 
    \[
    \frac{\ell_B(\{U_j\})}{\sqrt{t}} \leq \sigma_{\min} \left (\begin{bmatrix}
        U_1 & \dots & U_t 
    \end{bmatrix}
    \right) \leq \ell_B(\{U_j\}).
    \]
\end{lemma}

\begin{proof}
Let $\alpha \in \R^{tm}$ be a unit vector and write $\alpha = \alpha_1  \oplus \dots \oplus \alpha_t$. Intuitively, $\alpha_j$ records the entries of alpha that are coefficients of the columns of $U_j$ in the product $\begin{bmatrix}
        U_1 & \dots & U_t 
    \end{bmatrix} \alpha.$
Since $\alpha$ is a unit vector, there must exist some index $j$ such that $\|\alpha_j\| \geq \frac{1}{\sqrt{t}}$. From this we obtain
\[
\left\| \Pi^\perp_{-j} U_j \alpha_j \right\| = \left\| \Pi^\perp_{-j} \begin{bmatrix}
        U_1 & \dots & U_t 
    \end{bmatrix} \alpha \right\| \leq \left\|\begin{bmatrix}
        U_1 & \dots & U_t 
    \end{bmatrix} \alpha \right\|.
\]
Using the fact our lower bound on the norm of $\alpha_j$ and the fact that $\ell_B (\{U_j\})$ lower bounds the least singular value of $\Pi^\perp_{-j} U_j$, we obtain 
\[
\frac{\ell_B(\{U_j\})}{\sqrt{t}} \leq \left\| \Pi^\perp_{-j} U_j \alpha_j \right\|
\]
from which the desired lower bound follows. The proof of the upper bound is straightforward.
\end{proof}

\begin{lemma}\label{lem:perturbed-sing-value-basic}
Let $V \in \R^{n \times k}$ be an arbitrary matrix with $k\le n$, and let $\tilde{V}$ be a $\rho$ perturbation. Suppose $\alpha_i$ are some scalars such that $|\alpha_i| \ge \delta$ for all $i \le k$. Then for any $h \in (0, 1/2)$,
\[ \Pr \left[ \sigma_{\lceil k/2\rceil} \left( \tilde{V} \diag(\alpha) \right) < \frac{h\rho\delta}{4\sqrt e} \right] \le \exp(-\frac{1}{8} kn \log(1/h)).\]
\end{lemma}

\begin{proof}
Let us denote $W = \tilde{V} \diag(\alpha)$, for convenience. Set $q=\lceil k/2\rceil-1$ and $a=h\rho\delta/(4\sqrt e)$. If $\sigma_{q+1}(W)<a$, there is a set $J$ of $q$ columns such that $\|\Pi_J^\perp W\|_F\le ka$, where $\Pi_J^\perp$ projects onto the orthogonal complement of their span.\footnote{Here, we are using the well known connection between the low rank error and an approximation via columns~\cite{GuruswamiSinop}.}

Now fix a set $J$ of $q$ columns and condition on these columns. Project each remaining column onto an $(n-q)$-dimensional subspace orthogonal to their span, using orthonormal coordinates. The Gaussian vector formed by concatenating these projected columns has $D=(n-q)(k-q)$ coordinates and covariance at least $(\rho\delta)^2I_D$. Its Euclidean norm is at most $\|\Pi_J^\perp W\|_F$, so the event $\|\Pi_J^\perp W\|_F\le ka$ places this vector in a ball of radius $ka$. Bounding its probability by the Gaussian density times the volume of the ball gives
\[
\Pr[\|\Pi_J^\perp W\|_F\le ka\mid W_J]
\le\left(\frac{\sqrt e\,ka}{\rho\delta\sqrt D}\right)^D
\le(h/2)^D,
\]
where $k/\sqrt D\le2$. Since $D\ge k$ and $D\ge nk/4$, a union bound over at most $2^k$ choices of $J$ gives
\[
2^k(h/2)^D\le h^D\le\exp(-\tfrac18 nk\log(1/h)).
\]
\end{proof}

\end{document}